%% file: main.tex
\documentclass{article}

\usepackage[a4paper]{geometry}

\usepackage{authblk}
\usepackage{graphicx}
\usepackage{amsmath}
\usepackage{amssymb}
\usepackage[linesnumbered,ruled,vlined]{algorithm2e}
\usepackage{xspace}
\usepackage{color} 
\usepackage{url}
\usepackage[hidelinks]{hyperref}
\usepackage[utf8]{inputenc}
\usepackage{amsthm}

\newcommand{\Output}[1]{{\bf output} #1}
\SetKwProg{Fn}{Function}{:}{}
\SetKwFunction{BK}{BK}
\SetKwFunction{GCE}{GraphCliqueEnum}
\let\oldnl\nl % Store \nl in \oldnl
\newcommand{\nonl}{\renewcommand{\nl}{\let\nl\oldnl}}
\newcommand{\EG}[1]{\mathcal{E}_{#1}}
\newcommand{\etal}{\emph{et al.}\xspace}
\renewcommand{\P}{\mathcal{P}}
\renewcommand{\l}{\ell}
\newcommand{\lm}{\ell_{max}}
\newcommand{\lnm}{\ell_{\neg max}}
\newcommand{\N}{\Gamma_{G}}
\newcommand{\Nt}{\Gamma_{G_t}}
\newcommand{\Nx}[1]{\Gamma_{G_{#1}}}
\newcommand{\NC}{\mathcal{N}_C}

\newtheorem{definition}{Definition}
\newtheorem{theorem}{Theorem} % [section]
\newtheorem{lemma}{Lemma} % [section]

\begin{document}

%% --------------------------------------------------------------------
% Title
%% --------------------------------------------------------------------

\title{Faster maximal clique enumeration in large real-world link streams}

\author[1]{Alexis Baudin}
\author[1]{Clémence Magnien}
\author[1]{Lionel Tabourier}
\affil[1]{Sorbonne Université, CNRS, LIP6, F-75005 Paris}

\date{}

\maketitle

%% --------------------------------------------------------------------
% Content
%% --------------------------------------------------------------------

\begin{abstract}
  \input{abstract}

\end{abstract}

\input{introduction}

\input{definitions}
\input{related}

\input{algorithm}

\input{analysis}
\input{experiments}

\input{conclusion}

%% --------------------------------------------------------------------
% Acknowledgements
%% --------------------------------------------------------------------

\section*{Acknowledgements}

This work is financed by the ANR (French National Agency of Research) through the ANR FiT LabCom.
The authors thank the SocioPatterns\,\footnote{\url{http://www.sociopatterns.org}}, Konect~\cite{data-dnc} and SNAP~\cite{snapnets} communities, for making their datasets available.

%% --------------------------------------------------------------------
% Bibliography
%% --------------------------------------------------------------------

\bibliography{biblio}
\bibliographystyle{abbrv}

\end{document}

%% file: abstract.tex
%
% Context
%  
Link streams offer a good model for representing interactions over time.
They consist of links $(b,e,u,v)$, where $u$ and $v$ are vertices interacting during the whole time interval $[b,e]$.
In this paper, we deal with the problem of enumerating maximal cliques in link streams.
A clique is a pair $(C,[t_0,t_1])$, where $C$ is a set of vertices that all interact pairwise during the full interval $[t_0,t_1]$.
It is maximal when neither its set of vertices nor its time interval can be increased.
%
% State of the art 
%
Some main works solving this problem are based on the famous Bron-Kerbosch algorithm for enumerating maximal cliques in graphs.
%
% Method
%
We take this idea as a starting point to propose a new algorithm
which matches the cliques of the instantaneous graphs formed by links existing at a given time $t$ to the maximal cliques of the link stream.
%
% Experiments
%
We prove its correctness and compute its complexity, which is better than the state-of-the art ones in many cases of interest.
We also study the output-sensitive complexity, which is close to the output size, thereby showing that our algorithm is efficient.
To confirm this, we perform experiments on link streams used in the state of the art, and on massive link streams, up to 100 million links.
%
% Results
%
In all cases our algorithm is faster, mostly by a factor of at least 10 and up to a factor of $10^4$.
Moreover, it scales to massive link streams for which the existing algorithms are not able to provide the solution.

%% file: introduction.tex
\section{Introduction}

% motivation about temporal networks
% 
The analysis of real-world interaction networks has recently made significant progress as the field evolved from static to dynamic representations.
Indeed, the availability of temporal data, as well as the development of tools to describe and analyze them, have revealed the importance of the events' temporality.
They allow understanding the structure and functioning of complex interacting systems such as computer networks~\cite{miorandi2010}, online social networks~\cite{zhao2012multi}, human communication networks~\cite{peruani2011directedness}, or mobility networks~\cite{bajardi2011human,jacoby2016emerging}.

% motivation about link stream formalism
% 
Static representations of interacting systems are usually based on graphs. Concerning dynamical networks, several formalisms have been developed at the same time.
Some choose to represent a dynamical network as a sequence of static images of equal duration -- see for instance~\cite{li2017fundamental}.
This allows to use standard graph tools, but also raises issues about the existence of an adequate time window of analysis~\cite{hossmann2009contacts,leo2019non}.
In the same vein, \textit{time varying graphs}~\cite{casteigts2012time} represent a dynamical network as a graph where links have attributes corresponding to the time-stamps at which they exist in the network.
Some use a representation -- often simply called \textit{temporal network}~\cite{holme2012temporal} -- which aims at gathering the different existing models of dynamical interacting systems, without making a strict choice of formalism.
In this paper, we use link streams~\cite{latapy2018stream}, which associate a time interval to each interaction.
It proposes a rigorous formalism and allows accounting for the temporal aspect of the data without any parameter or restricting choice of timescale. 
Yet, it is simple to go from one mode of representation to another.

% cliques, communities
% 
Listing all maximal cliques in a graph is known to be NP-hard.
However, this problem is of utmost importance to analyze real-world graphs because it reveals high-density subgraphs and is thus a keystone to understand their structure.
For instance, clique mining is used to detect relevant dense subgraphs~\cite{gibson2005discovering,fratkin2006motifcut}, define communities~\cite{palla2005uncovering,baudin2022clique} or compress graphs~\cite{buehrer2008scalable}.
Efficient listing algorithms have been designed for large graphs representing real-world interaction networks.
In particular, the Bron-Kerbosch algorithm~\cite{bron1973} is used on large and sparse instances, especially by using the pivot improvement, which allows to considerably reduce the search space.
Currently, one of the most efficient versions of this algorithm is the pivot proposed by Tomita~\etal~\cite{tomita2006} and its implementation by Eppstein~\etal~\cite{eppstein2010}, which uses an adequate node ordering method and an efficient implementation for large sparse real-world graphs.

In link streams, cliques are sets of vertices interacting with all others during a time interval, and maximal cliques are maximal both in number of nodes and in duration.
Their listing has attracted interest in the last few years because it brings a powerful analysis tool to the field. To the best of our knowledge,
the first algorithm to achieve this task was proposed by Viard~\etal~\cite{viard2016} and later improved in~\cite{viard2018}.
In the meantime, Himmel \etal proposed a new version based on the adaptation of the Bron-Kerbosch algorithm to a dynamical context~\cite{himmel2017adapting}, which was later generalized to the notion of $k$-plex  by Bentert~\etal~\cite{bentert2019}.
Depending on the experimental settings, one or the other method can be faster.
However, these methods are still limited to relatively small dynamical networks, not allowing the enumeration on very large datasets.
Therefore, there is a need for more efficient algorithms to list maximal cliques in link streams.

% contributions
% 
The contributions of this paper are the following:
\begin{itemize}
\item[-] We propose a new algorithm for listing maximal cliques in link streams, which scales to massive real-world datasets,
\item[-] We analyze the complexity of this algorithm and provide two complexities: one that depends on the input and an output-sensitive complexity that is close to the output size, thereby showing that our algorithm is efficient.
  We also show that the memory requirements are close to optimal.
\item[-] We show experimentally that it significantly outperforms the state of the art and allows enumerating
  cliques in networks that are two orders of magnitude larger than what was previously feasible in the time and memory limits of the protocol,
  which allows us to step up from link streams with less than half a million links to link streams with more than 100 million links.
\item[-] Finally, we provide two implementations of the algorithm: one in \texttt{Python}, used for comparison to the state of the art and another in \texttt{C++}, which is the most efficient one currently available and also allows a parallel enumeration. 
  The code is publicly available\,\footnote{\url{https://gitlab.lip6.fr/baudin/maxcliques-linkstream}}.
\end{itemize}

% plan
% 
The rest of this paper is organized as follows.
Section~\ref{sec:definition} gives the basic definitions and notations that we use throughout the paper.
Section~\ref{sec:related} presents work in the literature related to maximal clique enumeration in link streams.
Then, our new algorithm is presented in Section~\ref{sec:algorithm}.
In Section~\ref{sec:analysis}, we analyze this algorithm, giving a proof of correctness as well as two theoretical complexities: one as a function of input parameters and one as a function of output parameters of the enumeration.
Finally, in Section~\ref{sec:experiments}, we make an extensive experimental assessment of the performances of this algorithm, and briefly show the results of the parallel implementation.

%% file: definitions.tex
\section{Basic definitions and notations}
\label{sec:definition}

\subsection{Cliques in a graph}

We briefly remind a few classic definitions in the context of graphs that are useful for this work.
We only consider simple undirected graphs.
Such a graph is defined as a pair $G = (V,E_G)$, where $V$ is a set of vertices, and $E_G$ a set of edges;
the edges of $E_G$ are of the form $\{x,y\}$, where $x,y \in V$ and $x \neq y$.
For a given vertex $u \in V$, the neighborhood of $u$, denoted $\N(u)$, is the set of vertices adjacent to $u$ in $G$, that is to say $\N(u) = \{v \in V \ | \ \{u,v\} \in E_G \}$.

A clique $C$ of a graph $G = (V,E_G)$ is a set of vertices of $V$ which are all connected to each other, or more formally
$$\forall u,v \in C \textit{ with } u \neq v, \{u,v\} \in E_G.$$
A clique is \emph{maximal} if it is not included in any other.
If $C$ is a clique of a graph $G$, then the \emph{neighborhood} of $C$ in $G$, denoted $\NC$, is the set of vertices which are neighbors of all vertices of $C$ in $G$: $\NC = \underset{v \in C}{\bigcap} \N(v).$

\subsection{Cliques in a link stream}\label{subsec:cliques_ls}

We recall here a few definitions on link streams which are needed for the rest of this paper.

\begin{definition}[Link stream]
  A link stream is a triplet $L = (T,V,E)$ where $T=[\alpha, \omega] \subset \mathbb{R} $ is a time interval, $V$ a set of nodes and $E \subseteq T \times T \times V \times V$ a set of links such that for all links $(b,e,u,v)$ in $E$, the beginning and ending times of the link $b$ and $e$ are such that $e \geq b$. We call $e-b$ the duration of the link. 
\end{definition}

In the rest of this article, the link streams that we use are undirected, \textit{i.e.} there is no distinction between $(b,e,u,v) \in E$ and $(b,e,v,u) \in E$. 
They are also simple, \textit{i.e.} for all $(b,e,u,v)$ in $E$, $u \neq v$ and for all $(b,e,u,v)$ and $(b',e',u,v)$ in E, $[b,e] \cap [b',e'] = \emptyset$.

A clique of a link stream $L=(T,V,E)$ is defined as follows:

\begin{definition}[Clique of a link stream]
  A clique is a pair $(C, [t_0,t_1])$
  where  $t_0,t_1 \in T$ are respectively called the start and end times of the clique, and $C \subseteq V$ is the set of vertices in the clique, with $|C| \geq 2$.
  Each pair of vertices of $C$ is connected by a link existing during the whole interval~$[t_0,t_1]$. 
  
  Formally, $(C, [t_0,t_1])$ is such that:
  $$\forall u,v \in C \textit{ with } u \neq v, \ \exists (b,e,u,v) \in E \text{ s.t. } [t_0,t_1] \subseteq [b,e].$$
\end{definition}

Note that for the sake of simplicity, we use the term clique both to designate a clique $C$ in a graph and a clique $(C, [t_0,t_1])$ in a link stream, while these objects are different in nature.
This is because in general the context allows removing the ambiguity on the kind of object considered.
As for a clique in a graph, a clique in a link stream can contain other cliques. 
During the exhaustive enumeration, we are only interested in those that are maximal.
In a link stream, the notion of maximality applies both to time and to vertices.
The following definitions formalize this maximality.
In this paper, we are interested in cliques that are maximal both in terms of time and vertices.
We start by defining these two notions that will be used in the description of our algorithm before defining maximal cliques formally.

\begin{definition}[Time-maximal clique]
  A time-maximal clique $(C,[t_0,t_1])$ is a clique which cannot be extended in time:  there is no clique
  $(C,[t_0',t_1'])$ with $[t_0,t_1] \subsetneq [t_0',t_1']$.
  
\end{definition}

\begin{definition}[Vertex-maximal clique]
  A vertex-maximal clique $(C,[t_0,t_1])$ is a clique which cannot be extended in vertices: there is no clique  $(C',[t_0,t_1])$ with  $C \subsetneq C'$.
\end{definition}

\begin{definition}[Maximal Clique]
  A maximal clique $(C,[t_0,t_1])$ in a link stream is a clique which is time-maximal and vertex-maximal.
\end{definition}

To illustrate these definitions, the left panel of Figure~\ref{fig:ls-cliques} shows a link stream with its maximal cliques.

\begin{figure}[!hbt]
  \begin{minipage}{0.7\linewidth}
    \centering
    \includegraphics[width=0.7\linewidth]{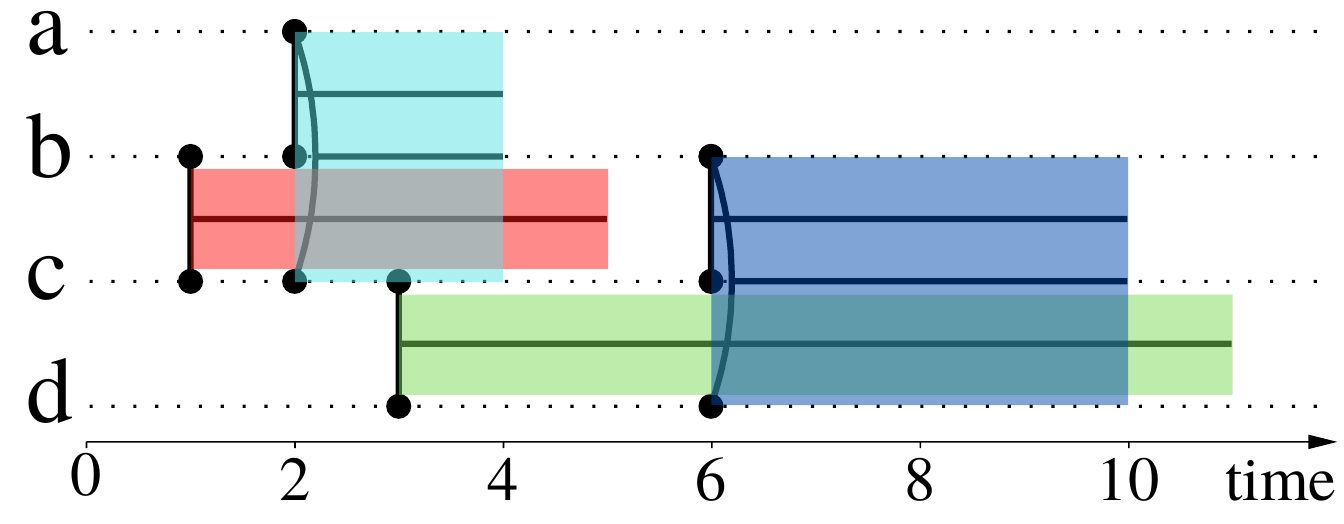} 
  \end{minipage}
  \begin{minipage}{0.25\linewidth}
    \includegraphics[width=0.8\linewidth]{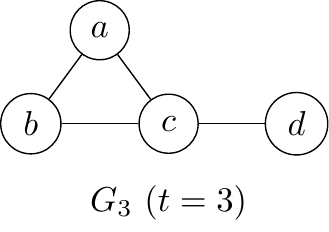} 
  \end{minipage}
  \caption{
    {\bf Left:} a link stream with  interaction time on the abscissa and vertices on the ordinate.
    For example, there is a link between $b$ and $c$ during the time interval $[1,5]$. The maximal cliques are represented in color, \textit{e.g.}, on the interval $[2,4]$, the three vertices $a,b,c$ are linked together, and form a maximal clique $(\{a,b,c\},[2,4])$.  
    {\bf Right:} the instantaneous graph $G_3$ of this link stream at~$t = 3$.}
  \label{fig:ls-cliques}
\end{figure}

\subsection{Instantaneous graph of a link stream}

We now give a few definitions that characterize a link stream $L = (T,V,E)$ at a given time $t \in T$.
The edges existing at $t$ can be seen as the edges of a graph that we call the instantaneous graph of $L$ at time $t$. 
This graph contains the vertices in $V$ and the edges existing at time $t$,
or more~formally:

\begin{definition}[Instantaneous graph $G_t$ associated to a link stream at time $t$]
  Given a ti\-me $t\in T$, the instantaneous graph $G_t$ at time $t$ associated to the link stream $L$ is the graph $G_t = (V,E_{G_t})$ such that
  $$E_{G_t} = \{\{u,v\} \ | \ \exists (b,e,u,v) \in E, \ t \in [b,e]\}.$$
\end{definition}

To avoid any confusion, we use the term \textit{edge} for the elements of $ E_G $ in a graph $ G=(V,E_G)$ and \textit{link} for the elements of $E$ in a link stream $ L=(T,V,E)$.
The edges of $G_t$ are induced by the links of $L$, the corresponding link thus has an end time in $T$.
We formalize below the notion of end time associated with an edge of $G_t$  and the final time of a clique of $G_t$ that follows.

\begin{definition}[End time $\EG{t}(u,v)$ of an edge $\{u,v\}$ of ${G_t}$]
  Let $\{u,v\} $ be an edge of ${G_t}$. By definition of $G_t$, there is a unique link $(b,e,u,v) \in E$ such that $t \in [b,e]$. We call $e$ the end time of the edge $\{u,v\}$ in $G_t$ and denote it $\EG{t}(u,v)$.
\end{definition}

\begin{definition} [Final time $\EG{t}(C)$ of a clique $C$ of $G_t$]
  Let $C$ be a clique of $G_t$. The final time $\EG{t}(C)$ of clique $C$ is the minimum of the end times of the edges of  $C$ in $G_t$. Formally:
  $$\EG{t}(C) = \min_{u,v \in C} \{ \EG{t}(u,v) \}.$$ 
\end{definition}

For instance, on the right panel of Figure~\ref{fig:ls-cliques}, the instantaneous graph at time $t=3$ is the graph $G_3 = (\{a,b,c,d\}, E_3)$, with $E_3 = \{\{a,b\},\{a,c\},\{b,c\},\{c,d\}\}$.
The end times of its edges are $\EG{3}(a,b) = \EG{3}(a,c) = 4$, $\EG{3}(b,c) = 5$ and $\EG{3}(c,d) = 11$.
$G_3$ contains the clique $\{a,b,c\}$, which also exists in $G_2$ and in $G_4$, and whose final time is $4$, corresponding to the minimum of the end times of the three edges.

%% file: related.tex
\section{Related work}
\label{sec:related}

\subsection{Maximal cliques in graphs}
\label{sec:BKStatic}

We first describe the Bron-Kerbosch algorithm~\cite{bron1973}, denoted BK, which is widely used to detect maximal cliques in graphs representing real-world networks, and serves as a basis for the design of the algorithm that we propose in this work.
It is formally described in Algorithm~\ref{algo:BKstatic}.

\begin{algorithm}[!hbt]
  \DontPrintSemicolon 
  \KwIn{Graph $G = (V,E)$}
  \KwOut{All maximal cliques of $G$ without duplicates}
  \BK{$\emptyset$, $V$, $\emptyset$} \;
  \Fn{\BK{$R$, $P$, $X$}}{
    \If{$P \cup X = \emptyset$} { \label{line:BKstatic:maxtext_start}
      \Output{$R$ maximal clique} \;
    }
    \For {$u \in P$} { \label{line:BKstatic:Foru} 
      \BK{$R \cup \{u\}$, $P \cap \N(u)$, $X \cap \N(u)$} \label{line:BKstatic:RC} \;
      $P \gets P \setminus \{u\}$ \;
      $X \gets X \cup \{u\}$ \;
    }
  }
  \caption{Bron-Kerbosch algorithm (without pivot) on a graph $G$}
  \label{algo:BKstatic}
\end{algorithm}

Broadly speaking, it is a recursive backtracking algorithm.
The central idea is that each call maintains a clique $R$ and the neighbors of all the vertices of $R$, which are distributed in two sets: $P$ and $X$. 
Each vertex of $P \cup X$ can potentially expand the clique $R$. 
$P$ corresponds to the candidate nodes that are actually used to expand $R$, while $X$ corresponds to the vertices that are forbidden for expanding $R$ to avoid the enumeration of duplicate cliques. 
Note that a clique $R$ is maximal if and only if there is no vertex which is a neighbor of all its vertices, that is to say if and only if~$P \cup X = \emptyset$.

In terms of complexity, this version of the BK algorithm makes exactly one recursive call per clique of the graph, maximal or not. Note that there can be many such cliques, as a given maximal clique of size $q$ contains $2^q$ sub-cliques.
To reduce this search complexity, Bron and Kerbosch have introduced the idea of selecting a pivoting vertex to prune the recursive call tree.
Indeed, by choosing a vertex $p \in P \cup X $, we do not need to make recursive calls for the vertices in $\N(p)$ because any maximal clique must include either $p$ or a node which is not in $ \N(p) $. So we do not miss any maximal clique by cutting these branches of the tree.
The corresponding modification is:

\begingroup
\setlength{\interspacetitleruled}{0pt} 
\setlength{\algotitleheightrule}{0pt}
\begin{algorithm}[!hbt]
  \DontPrintSemicolon
  \nonl \emph{Replace Line~\ref{line:BKstatic:Foru} of Algorithm~\ref{algo:BKstatic} by:} \;
  \nonl \mbox{}\phantom{for} Choose pivot $p \in P \cup X$ \;
  \nonl \mbox{}\phantom{for} {\bf for} $u \in P \setminus \N(p) $ {\bf do} \;
\end{algorithm}
\endgroup

The strategies to select the pivot and achieve an efficient pruning have been discussed in various works, \textit{e.g.}~\cite{koch2001enumerating}.
In particular, the one proposed by Tomita \etal~\cite{tomita2006} maximizes
$|P \cap \N(p)|$ over all $p \in P \cup X$ and this guarantees that the worst-case time complexity is in $ \mathcal{O}\left(3^{n/3}\right)$ with $n=|V|$.
Eppstein \etal~\cite{eppstein2010} use this strategy and also propose to make the first call of the {\tt BK} function on each vertex according to the degeneracy ordering.
The underlying idea is that when the vertices are processed in a random order, $|P|$ can be as large as the maximum degree, while when using the degeneracy ordering, $|P|$ cannot be larger than the graph degeneracy $\delta$, which is smaller than the maximum degree and usually low in graphs representing real-world networks.
They prove that their method ensures a time complexity in $\mathcal{O}\left(\delta \cdot n \cdot 3^{\delta/3}\right)$.
Combined with adequate data structures, their implementation is currently considered to be one of the most efficient versions available of the BK algorithm.

\subsection{Maximal cliques in  link streams}
\label{sec:relatedtemporal}

The question of listing maximal cliques in dynamical networks has emerged relatively recently as an important research question to describe the structure of interaction data with temporal information.
This problem has been considered in different ways: some see a dynamical network as a sequence of graph snapshots, which leads to tackle the problem as the evolution of cliques in graphs through time.
This is for instance the point of view adopted in~\cite{sun2017mining,das2019incremental}.
Other works model the dynamical network as a link stream, thus acknowledging the intrinsically temporal aspect of the clique itself and thus look at it as an object which should be redefined in the dynamical context~\cite{viard2016,himmel2017adapting,viard2018,bentert2019}.
We focus on this second point of view.

In this paper, we consider works that define cliques in link streams.
This issue has been addressed with slightly different formalisms: some consider the link streams formalism~\cite{viard2016,viard2018} and others use the temporal networks one~\cite{himmel2017adapting,bentert2019}.
All these works nonetheless consider cliques as a set of vertices interacting during a given time interval, as discussed in Section~\ref{subsec:cliques_ls}.
Both representations contain the same information, and it is simple to go from one representation to the other, making it possible to compare the efficiency of the different methods.

Viard \etal~\cite{viard2016} proposed the first algorithm to list maximal cliques in link streams. 
In that work, links do not have duration, but the cliques are themselves parameterized with a value $ \Delta $, and thus called $ \Delta $-cliques.
In this work $\Delta$-cliques are defined such that all pairs of nodes in the $ \Delta $-clique interact at least once during each sub-interval of duration $ \Delta $.
Their algorithm has since then been both simplified and extended to the more general formalism of link streams that allow link duration~\cite{viard2018}.
The idea underlying this algorithm is to extend either in terms of time or nodes a clique in a link stream, starting from a specific link of the stream.
While this method indeed allows finding all maximal cliques in the link stream, it relies on memoization in order to decide whether a given clique has already been processed and this induces a high memory usage, which is prohibitive in many cases.

In parallel, Himmel \etal~\cite{himmel2017adapting} proposed another algorithm to enumerate $ \Delta $-cliques in link streams (named temporal graphs in these papers).
It adapts the BK algorithm to this dynamical setting, and implements different pivoting strategies.
This version regularly outperforms the algorithm in~\cite{viard2016}, especially on larger $ \Delta $ values. 
The results are more contrasted when compared to the algorithm in~\cite{viard2018}: depending on the experimental settings and the data under study, one or the other method may be more efficient.
More recently, Bentert \etal~\cite{bentert2019} have developed a generalization of this algorithm to list all maximal $\Delta $-$k$-plexes in link streams. 
A $ \Delta $-$k$-plex in a link stream is defined by analogy with $k$-plexes in a graph, which is a maximal subgraph of size $s$ such that any vertex in the $k$-plex is adjacent to at least $ s-k $ vertices in the subgraph.
So a $ \Delta $-$k$-plex in a link stream is a subset of $s$ nodes and a time interval $ \Delta $ in which each vertex interacts with at least $s-k$ vertices of the $ \Delta $-$k$-plex
during each sub-interval of duration $ \Delta $.
Note that a  $\Delta $-$1$-plex ($k=1$) is equivalent to a $ \Delta $-clique, which allows comparing this algorithm to the others described above.
Here also, the algorithm is based on the structure of the BK algorithm, but it contains a few implementation improvements on~\cite{himmel2017adapting}, which makes it more efficient in terms of practical computation times.
These two methods obtain running time bounds depending on a temporal variant of the degeneracy of the input graph, similarly to Eppstein~\etal~\cite{eppstein2010} in the static case.

Banerjee and Pal have proposed the notion of maximal $(\Delta,\gamma)$-cliques~\cite{banerjee2022efficient}, which is a generalization of maximal $\Delta$-clique: a $(\Delta, \gamma)$-clique is defined in the same way as a $\Delta$-clique, but each link must appear at least $\gamma$ times in each sub-interval of size $\Delta$.
Note that for $\gamma=1$, their definition is equivalent to maximal $\Delta$-cliques.
Molter~\etal also proposed a generalization of maximal clique enumeration in link streams~\cite{molter2021isolation}, by enumerating \emph{isolated} maximal cliques, that are cliques with few edges connected to the rest of the temporal network. The isolation is quantified by a factor $c$ that is the maximal number of links between the clique and the rest of the network. We can recover the
definition of maximal clique with an isolation condition that makes all cliques isolated (\textit{e.g.}
take $c=n^2$).

Finally, some methods which were originally designed to update cliques in graphs that evolve with time may be exploited in the context of clique enumeration in link streams.
In particular, Das \etal~\cite{das2019incremental} apply the Tomita \etal~\cite{tomita2006} clique enumeration method in this context.
This is done by enumerating the cliques that contain at least one edge that is new at the current time.
While the formal problem is different from ours, this method can be directly adapted to our problem, as we will see in Section~\ref{subsec:generalstructure}.

%% file: algorithm.tex
\section{Algorithm}
\label{sec:algorithm}

The algorithm that we propose is based on an adaptation of the BK algorithm to a dynamical setting, following the same logic as the one proposed in Himmel \etal~\cite{himmel2017adapting}.
We use the BK procedure to enumerate exactly once all sets of vertices $C$ that form a time-maximal clique $(C,I)$, and then filter these time-maximal cliques by keeping only those that are vertex-maximal. 
The main difference of our approach is that we work on neighborhoods which are limited to the interactions actually occurring at time $t$ and not to the complete link stream.

In what follows, we consider a simple link stream $ L=(T,V,E) $.
We describe the general structure of our algorithm, then we detail the two main points: first, enumerating the time-maximal cliques via the enumeration of cliques of instantaneous graphs $G_t$,
second, testing their vertex-maximality to filter out those which are not.
Finally, we improve the computation times
by pruning the search tree with a \textit{pivot} that is similar to the one used in~\cite{himmel2017adapting}.

\subsection{General structure of the algorithm}
\label{subsec:generalstructure}

The algorithm that we propose first enumerates exactly once each time-maximal clique and then filters the ones which are vertex-maximal.
For the first step, we use the equivalence given by Lemma~\ref{thm:timemax}.
\begin{lemma}[Time-maximality of a clique]
  \label{thm:timemax}
  $(C,\left[t_0,t_1\right])$ is a time-maximal clique if and only~if:
  \begin{description}
  \item [$(i)$] $C$ is a clique of $G_{t_0}$;
  \item [$(ii)$] There is an edge $\{u,v\}$ in $G_{t_0}$ with $u,v \in C$ that originates from a link $(t_0,e,u,v) \in E$ that begins at~$t_0$;
  \item [$(iii)$] $t_1 = \EG{t_0}(C)$.
  \end{description}
\end{lemma}

\begin{proof} 
  We suppose that $(C,\left[t_0,t_1\right])$ is a time-maximal clique.
  Then $C$ is trivially a clique of $G_{t_0}$ as    all pairs of nodes of $C$ are connected at time $t_0$.
  Moreover, Viard \etal~\cite{viard2018} proved in their Lemma~3 that there must exist an edge $\{u,v\}$ with $u,v \in C$ that originates from a link $(t_0,e,u,v) \in E$ that begins at $t_0$.
  Indeed, if all the links in the clique begin before $t_0$, its time interval can be extended to the left, and it is therefore not time-maximal.
  Finally, as the clique  $(C,\left[t_0,t_1\right])$ cannot be extended in time, it implies that the final time associated with the clique must be the maximum time when all links are present in the link stream, in other words $\EG{t_0}(C) = t_1$.
  
  Reciprocally, if $C$ is a clique of $G_{t_0}$ and $t_1 = \EG{t_0}(C)$, it means that all edges in $C$ originates from links of the link stream that are present over the whole interval $\left[t_0,t_1\right]$, so $(C,[t_0,t_1])$ is a clique of $L$.
  Moreover,
  the time interval of this clique is maximal,
  as (1) at least one link is no longer present after $t_1$;
  (2)
  if there is an edge $\{u,v\}$ with $u,v \in C$ that originates from a link $(t_0,e,u,v) \in E$ beginning at $t_0$, it means that at least one edge is not present before $t_0$.
  So the clique $(C,\left[t_0,t_1\right])$ is time-maximal. 
\end{proof}	

\begin{figure}[!h]
  \begin{minipage}{0.7\linewidth}
    \centering
    \includegraphics[width=0.7\linewidth]{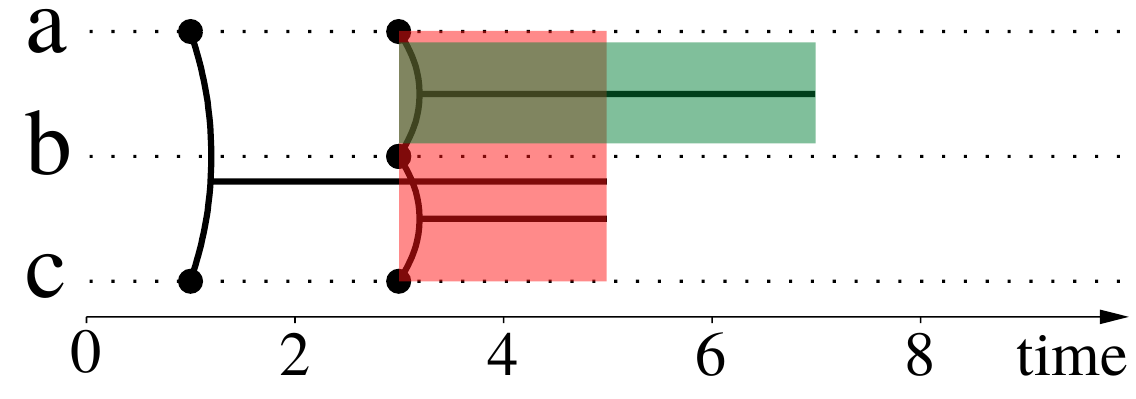}
  \end{minipage}
  \begin{minipage}{0.25\linewidth}
    \includegraphics[width=0.5\linewidth]{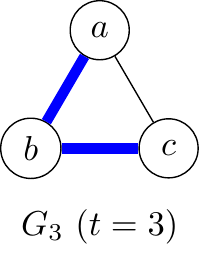}
  \end{minipage}
  \caption{
    A link stream with the two maximal cliques that start at $t=3$ in color, and its instantaneous graph $G_3$ with the edges that start at $t=3$ in blue. There are three time-maximal cliques that start at this instant, and each contains a new edge (in blue): $(\{a,b\},[3,7])$, $(\{b,c\},[3,5])$ and $(\{a,b,c\},[3,5])$. The clique $(\{b,c\},[3,5])$ is not vertex-maximal, while the two others are.}
  \label{fig:algo-example}
\end{figure}

We deduce from Lemma~\ref{thm:timemax} that we can enumerate the time-maximal cliques by going through $T$:
for each instant $t \in T$, we enumerate the time-maximal cliques that start at time $t$.
To do so, we enumerate each (graph) clique $C$ of $G_t$ that contains an edge $\{u,v\}$ of a link $(t,e,u,v) \in E$ starting at $t$.
Each of these graph cliques is then associated to the time-maximal clique $(C,[t,\EG{t}(C)])$ of the link stream.
The lemma ensures that we list all time-maximal cliques in this way.
Then, we only have to find those that are also vertex-maximal.
Figure~\ref{fig:algo-example} gives an illustration of this procedure at time $t=3$:
the time-maximal cliques that begin at $t=3$ are $(\{a,b\},[3,7])$, $(\{b,c\},[3,5])$ and $(\{a,b,c\},[3,5])$, and they all contain
at least one new link and their end time corresponds to the minimum of the end times of their links.
Then, only $(\{a,b\},[3,7])$ and $(\{a,b,c\},[3,5])$ are output, since $(\{b,c\},[3,5])$ is not vertex-maximal. 
This framework is summarized in Figure~\ref{fig:struct-algo}, and is implemented in Algorithm~\ref{algo:framework}.

\begin{figure}[!hbt]
  \centering
  \includegraphics[width=0.9\linewidth]{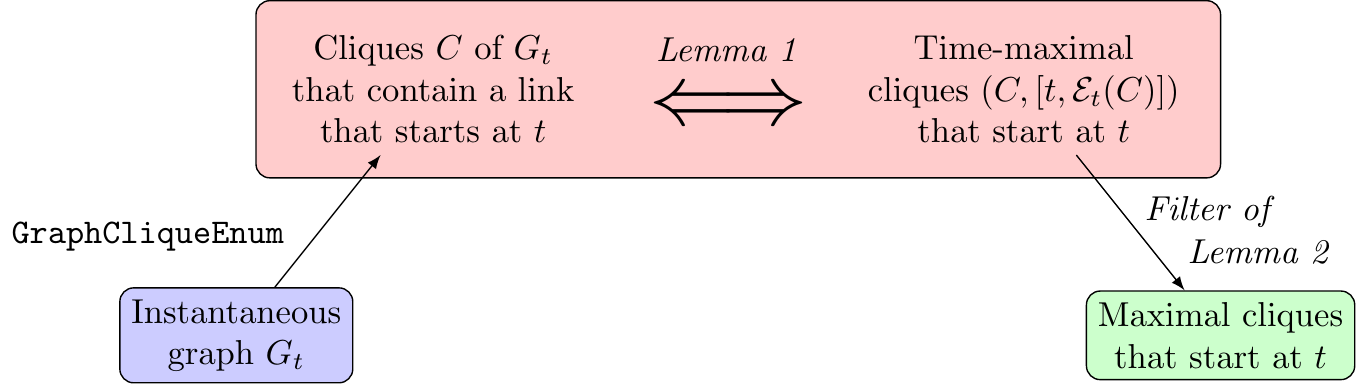}
  \caption{General structure of Algorithm~\ref{algo:framework}: for each time $t \in T$, it enumerates the maximal cliques that start at $t$.}
  \label{fig:struct-algo}
\end{figure}

\begin{algorithm}[!hbt]
  \DontPrintSemicolon
  \KwIn{A link stream $L=(T,V,E)$}
  \KwOut{All maximal cliques of $L$ without duplicates}
  \For {$t \in T$} { \label{line:framework:fort}
    $Cliques \gets \emptyset$ \; \label{line:framework:emptyCliques}
    $ForbidEdges \gets \emptyset$ \;
    \For (\tcp*[f]{Loop on links starting at time $t$}) {$(t,\_,u,v) \in E$} {
      \label{line:framework:fortuv}
      \tcp{Enumerating cliques in $G_t$ containing $\{u,v\}$:}
      $Cliques \gets Cliques \ \cup \ $\GCE{$\{u,v\}$, $\Nt(u) \cap \Nt(v)$, $\emptyset$, $ForbidEdges$, $t$} \;
      \label{line:framework:cliques}
      $ForbidEdges \gets ForbidEdges \cup \{\{u,v\}\}$ \; \label{line:framework:addforbedge}
    } \label{line:framework:endFor1}
    \For (\tcp*[f]{$\NC =$ neighborhood of $C$}) {$(C, \NC) \in Cliques$} {
      \label{line:framework:forC}
      \If (\tcp*[f]{Vertex maximality test}) {$\ \forall u \in \NC, \ \ \EG{t}(C \cup \{u\}) < \EG{t}(C)$} {
        \label{line:framework:maxtest}
        \Output{$\left(C,\left[t,\EG{t}(C)\right]\right)$}
        \label{line:framework:output}
      }
    } \label{line:framework:endForC}
  }

  \Fn{\GCE{$R$, $P$, $X$, $ForbidEdges$, $t$}}{ \label{line:staticenum:start}
    \Output {$(R,P\cup X)$} \tcp*{$P \cup X =$ neighborhood of $R$}
    \label{line:staticenum:output}
    $Q \gets \{u \in P \ | \ \exists \{u,v\} \in ForbidEdges, \ v \in R\}$ \;
    \label{line:staticenum:Q}
    \For {$u \in P \setminus Q$}{ \label{line:staticenum:forup}
      \GCE{$R \cup \{u\}$, $P \cap \Nt(u)$, $X \cap \Nt(u)$, $ForbidEdges$, $t$} \;
      \label{line:staticenum:RC}
      $P \gets P \setminus \{u\}$ \;
      \label{line:staticenum:Pu}
      $X \gets X \cup \{u\}$ \;
      \label{line:staticenum:Xu}
    }
  }
  \caption{Algorithm of maximal clique enumeration in link streams.}
  \label{algo:framework}
\end{algorithm}

We first describe the high-level features of Algorithm~\ref{algo:framework}. The details will be discussed in the rest of this section.
For each time $t \in T$ (Line~\ref{line:framework:fort}), the algorithm follows two steps:
\begin{itemize}
\item[-] Lines~\ref{line:framework:emptyCliques} to~\ref{line:framework:endFor1} enumerate exactly once each time-maximal clique beginning at $t$.
  To do this, the cliques of $G_t$ that contain an edge from a link starting at $t$
  are enumerated by processing every link $(t,\_,u,v) \in E$ that starts at $t$ (Line~\ref{line:framework:fortuv}).
  This enumeration is made by a call to the function \GCE (Line~\ref{line:framework:cliques}).
  Note the use of the set of edges $ForbidEdges$ that allows to avoid listing a same clique several times, which will be discussed in further details in Section~\ref{sec:staticenum}.
\item[-] Lines~\ref{line:framework:forC} to~\ref{line:framework:endForC} test whether the time-maximal cliques enumerated above are vertex-maximal or not. To do so, we use the neighborhood $\NC$
  of $C$ for each time-maximal clique $(C,[t_0,t_1])$.
  The maximality test (Line~\ref{line:framework:maxtest}) will be discussed in Section~\ref{sec:vertexmax}.
\end{itemize}

Note also that the iterations of the loop at Line~\ref{line:framework:fort} of the algorithm are independent of each other and can therefore be computed in parallel.
We describe this parallelization process and its results in Section~\ref{sec:experiments}.

\subsection{Clique enumeration in graphs $G_t$}
\label{sec:staticenum}

In this section, we discuss the procedure \GCE{$R, P, X, ForbidEdges, t$} which enumerates all graph cliques $C$ of $G_t$
satisfying $R \subseteq C \subseteq R \cup P$ and containing no edge of $ForbidEdges$.
It is called on Line~\ref{line:framework:cliques} of Algorithm~\ref{algo:framework} and detailed from Line~\ref{line:staticenum:start} to~\ref{line:staticenum:Xu}.
This procedure is a variant of the BK algorithm, described in Section~\ref{sec:BKStatic}: it is a backtracking recursive function which explores the neighborhood of the clique under construction to extend it.
Here, the clique under construction is $R$, its neighborhood is $P \cup X$, and the candidates for increasing $R$ without enumerating duplicates are the vertices of $P$.
The fact that $P \cup X$ is the neighborhood of clique $R$ is formally proved in Section~\ref{sec:validity} by Lemma~\ref{lemma:outputGCE}.
There are three major differences with the BK procedure described in Algorithm~\ref{algo:BKstatic} that should be noted:
\begin{itemize}
\item[-] each clique $R$ is output with its neighborhood $P \cup X$ (Line~\ref{line:staticenum:output}) because it is necessary for the vertex-maximality test of Line~\ref{line:framework:maxtest}, as discussed in Section~\ref{sec:vertexmax};
\item[-]  a clique is output whether it is maximal in $G_t$ or not  (Line~\ref{line:staticenum:output}): there is no test equivalent to Line~\ref{line:BKstatic:maxtext_start} of Algorithm~\ref{algo:BKstatic};
\item[-] the output cliques must not contain any edge of $ForbidEdges$;
  this is ensured by Line~\ref{line:staticenum:Q} which prevents to add a vertex $u$ to $R$ if it implies the addition of an edge $\{u,v\}$ of $ForbidEdges$  to the resulting clique $R \cup \{u\}$.
\end{itemize}

Using the set $ForbidEdges$ guarantees to list each clique at most once, as will be proved in Section~\ref{sec:analysis}.
Indeed, if a clique $C$ contains two edges $\{u_1,v_1\}$ and $\{u_2,v_2\}$, then it is in the set of cliques of $G_t$ which contain $\{u_1,v_1\}$, but also in the set of cliques which contain $\{u_2,v_2\}$.
Thus, if \GCE is called at $t$ on $\{u_1,v_1\}$ and $\{u_2,v_2\}$ without $ForbidEdges$, $C$ would be enumerated twice.
This idea is illustrated by the example in Figure~\ref{fig:cledges}:
\GCE is called first on $\{a,c\}$ and output cliques $\{a,c\}$, $\{a,b,c\}$ $\{a,c,d\}$ and $\{a,b,c,d\}$.
Then $\{a,c\}$ is added to $ForbidEdges$ and \GCE is called on $\{b,d\}$, outputting cliques $\{b,d\}$, $\{a,b,d\}$ and $\{b,c,d\}$ but not $\{a,b,c,d\}$ because it contains the edge $\{a,c\}$ which is in $ForbidEdges$.
Thus, all cliques containing at least one red edge are listed exactly once.

\begin{figure}[!hbt]
  \centering
  \includegraphics[width=0.25\linewidth]{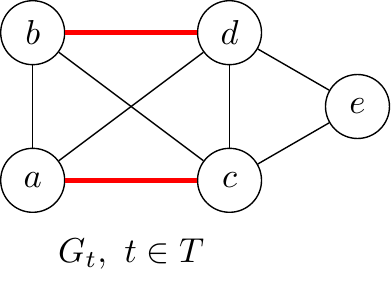}
  \caption{
    Example of a clique enumeration using $ForbidEdges$.
    An instantaneous graph $G_t$ of a link stream $L$ at time $t$ is represented.
    The thick red edges correspond to a link that begins at $t$, while the others correspond to links that have begun earlier.
    The clique $\{a,b,c,d\}$ contains the two new edges $\{a,c\}$ and $\{b,d\}$, which is why there is a need for the set $ForbidEdges$ to avoid enumerating it twice.
  }
  \label{fig:cledges}
\end{figure}

\subsection{Vertex maximality test of time-maximal cliques}
\label{sec:vertexmax}

In this section, we explain the vertex maximality test used in Line~\ref{line:framework:maxtest} of Algorithm~\ref{algo:framework}.
The following lemma allows filtering the vertex-maximal cliques among the time-maximal cliques.

\begin{lemma}[Vertex-maximality of a time-maximal clique]
  \label{thm:vertexmax}
  Let $(C,[t,\EG{t}])$ be a time-maximal clique and $\NC$ the neighborhood of clique $C$ in $G_t$   (\emph{i.e.} $\NC = \bigcap_{v \in C} \Nx{t}(v)$).
  Then
  $$ (C,[t,\EG{t}(C)]) \text{ is vertex-maximal } \Leftrightarrow \forall u \in \NC, \ \ \EG{t}(C \cup \{u\}) < \EG{t}(C).$$ 
\end{lemma}

In other words, a time-maximal clique $(C,[t,\EG{t}(C)])$ is vertex-maximal (and thus maximal) if and only if we cannot add any node to $C$ without reducing its final time.
Note that if the neighborhood of $C$ is empty, \textit{i.e.} $\NC = \emptyset$, then $(C,[t,\EG{t}(C)])$ is always vertex-maximal.

\begin{proof}
  Let $(C,[t,\EG{t}(C)])$ be a maximal clique and suppose
  that there exists some $u \in \underset{v \in C}{\bigcap} \Nx{t}(v)$ such that $\EG{t}(C \cup \{u\}) \geq \EG{t}(C)$.
  Then, by definition of $\EG{t}(C \cup \{u\})$, for any pair of nodes $x,y \in C \cup \{u\}, \ x \neq y$, there must be a link $(b,e,x,y) \in E$ such that $[t,\EG{t}(C)] \subseteq [b,e]$. In consequence, $(C \cup~\{u\}, [t,\EG{t}(C)])$ is a clique of $L$, which
  contradicts our hypothesis that $(C,[t,\EG{t}(C)])$ is maximal.
  This proves that a time-maximal clique $ (C,[t,\EG{t}(C)]) \text{ is vertex-maximal }$ implies that $\forall u \in~\underset{v \in C}{\bigcap} \Nx{t}(v), \ \ \EG{t}(C\cup~\{u\})  <~\EG{t}(C).$

  Reciprocally, suppose that $(C,[t,\EG{t}(C)])$  is not vertex-maximal, then by definition there must be a vertex $u \in  \underset{v \in C}{\bigcap} \Nx{t}(v)$ such that $(C \cup \{u\},[t,\EG{t}(C)])$ is a clique of $L$. This implies in turn that $ \EG{t}(C) \leq \EG{t}(C \cup \{u\})$.
  So, if $\forall u \in \underset{v \in C}{\bigcap} \Nx{t}(v), \ \ \EG{t}(C \cup \{u\}) < \EG{t}(C)$ then $(C,[t,\EG{t}(C)])$  is vertex-maximal.
\end{proof}

Therefore, the filtering test performed at Line~\ref{line:framework:maxtest} does correspond to the maximality test  of the above lemma.
For example, in Figure~\ref{fig:algo-example}, the time-maximal clique $(\{b,c\},[3,5])$ is not vertex-maximal, because we can add  vertex $a$ to it without reducing its final time. This is not the case for
$(\{a,b\},[3,7])$ and $(\{a,b,c\},[3,5])$ which are both time-maximal and vertex-maximal, and thus~maximal.

\subsection{Pivoting strategy to improve clique enumeration in $G_t$ graphs}
\label{sec:pivot}

As reported in Section~\ref{sec:BKStatic}, the BK algorithm in graphs is usually implemented with a pivot to prune the recursion tree.
A similar process can be implemented in the dynamical context to eliminate some recursive calls when enumerating time-maximal cliques, as proposed in~\cite{himmel2017adapting}. 
We adapt this improvement to the function \GCE, which corresponds to the following modification of Algorithm~\ref{algo:framework}:

\begingroup
\setlength{\interspacetitleruled}{0pt} 
\setlength{\algotitleheightrule}{0pt}
\begin{algorithm}[!hbt]
  \DontPrintSemicolon
  \nonl \emph{Replace Line~\ref{line:staticenum:forup} of Algorithm~\ref{algo:framework} by:} \;
  \nonl \mbox{}\phantom{for} Choose pivot $p \in P \cup X$ \;
  \nonl \mbox{}\phantom{for} $Del \gets \{u \in P \cap \Nt(p) \ | \ \EG{t}(R \cup \{u\}) = \EG{t}(R \cup \{u,p\})\}$ \;
  \nonl \mbox{}\phantom{for} {\bf for} $u \in (P \setminus Del) \setminus Q$ {\bf do}
\end{algorithm}
\endgroup

According to this procedure, vertices in the set $Del$ are not considered in the recursive calls.
Indeed,
not making recursive calls for these vertices does not remove any maximal clique from the enumeration.
This statement is formally proved by Theorem~\ref{thm:correctpivot}, in Section~\ref{sec:validity}.
This is justified by the observation that at a given time $t$, for any  maximal clique $(C,[t,\EG{t}(C)])$
such that $R\subseteq C \subseteq R \cup P$, and given a pivot $p\in P\cup X$, there is always a vertex $u \in P \setminus Del$ which allows extending $R$ towards $C$.
To observe that, there are two cases:
\begin{itemize}
\item[-] if there exists a vertex $u\in C$ that is not a neighbor of $p$, then $u \notin Del$ since $Del \subseteq \Nt(p)$, and $u \in P$ since $R \subseteq \Nt(p)$. Then, a recursive call with $u$ extends $R$ towards $C$  -- note that if $ p \in C$ then $u=p$ as $p \notin \Nt(p)$,
\item[-] else, all the vertices in $C$ are neighbors of $p$, which means on the one hand that $p \notin C$, and on the other hand that $C \cup \{p\}$ is a clique of $G_t$.
  Since $(C,[t,\EG{t}(C)])$ is maximal, this implies that $ \EG{t}(C \cup \{p\}) < \EG{t}(C)$.
  Therefore there must exist a vertex $u\in P\cap C$ such that
  $ \EG{t}(R \cup \{u\}) > \EG{t}(R \cup \{u,p\})\}$, {\em i.e.} $u$ does not belong to $Del$ and allows extending $R$ towards~$C$. 
\end{itemize}

Besides, we mentioned in Section~\ref{sec:BKStatic} that the Tomita \etal pivoting strategy for clique enumeration in graphs~\cite{tomita2006} chooses a pivot that maximizes the set of vertices to be removed from the candidates. 
Similarly, we select a pivot that allows to cut as many calls as possible: for each potential $p \in P \cup X$, we choose the one that maximizes $|Del|$.
It has been shown experimentally in~\cite{himmel2017adapting} that this choice does indeed achieve the most efficient pruning.
We evaluate in Section~\ref{sec:pivoting} the pruning efficiency by comparing the calculations made with and without a pivot.

%% file: analysis.tex
\section{Analysis}
\label{sec:analysis}

In this section, we show that Algorithm~\ref{algo:framework} is correct: it allows enumerating once and only once each maximal clique of a link stream.
Then, we study its complexity, first as a function of the global parameters of the link stream, then as a function of the output of the algorithm.

\subsection{Correctness}
\label{sec:validity}

To show that Algorithm~\ref{algo:framework} is correct, we need the following preliminary lemmas: Lemma~\ref{lemma:outputGCE} which characterizes the output of the calls to $\GCE$ in the algorithm, and Lemma~\ref{lemma:enum-timemax} which shows that all the maximal cliques starting at time $t$ are enumerated by the procedure from Lines~\ref{line:framework:emptyCliques} to~\ref{line:framework:endFor1}.
Finally, the correctness of the algorithm is proved in Theorem~\ref{thm:validity}.

% ****************
% Correct output * 
% ****************

\begin{lemma}
  \label{lemma:outputGCE}
  In Algorithm~\ref{algo:framework}, each call to $\GCE{R,P,X,ForbidEdges,t}$ satisfies:
  \begin{itemize}
  \item[-] $R$ is a clique of $G_t$;
  \item[-] $P \cup X = \underset{v \in R}{\bigcap} \Nt(v)$.
  \end{itemize}
\end{lemma}

\begin{proof}
  \GCE being a recursive function, we prove this lemma by induction on the tree of the recursive calls made by Algorithm~\ref{algo:framework}.
  \begin{description}
  \item [Initialization.] The first call of a series of recursive calls is made at Line~\ref{line:framework:cliques} of Algorithm~\ref{algo:framework}, and is of the form \GCE{$\{u,v\}, \Nt(u) \cap \Nt(v), \emptyset, ForbidEdges , t$}, for some $u,v \in V$. 
    At this point, $R = \{u,v\}$, and $u$ and $v$ are connected in $G_t$ because they come from the link $(t, \_, u, v)$ of $E$ (Line~\ref{line:framework:fortuv}).
    So $R$ is indeed a clique of $G_t$. Besides, $P = \Nt(u) \cap \Nt(v)$ and $X = \emptyset$, so $P \cup X = \underset{v \in R}{\bigcap} \Nt(v)$.
  \item [Induction.]
    Suppose \GCE{$R,P,X,ForbidEdges,t$} satisfies that $R$ is a clique of $G_t$ and $P \cup X = \underset{v \in R}{\bigcap} \Nt(v)$.
    Let us show that the recursive calls made at Line~\ref{line:staticenum:RC} induced by this call maintain these properties.
    On the one hand, $P \cup X = \underset{v \in R}{\bigcap} \Nt(v)$ is an invariant of the loop starting at Line~\ref{line:staticenum:forup}: indeed, after the operations of Lines~\ref{line:staticenum:Pu} and~\ref{line:staticenum:Xu}, $P \gets P \setminus \{u\}$ and $X \gets X \cup \{u\}$, so $P \cup X$ is not modified. 
    On the other hand, at a given loop iteration associated with vertex $u$, the recursive call is made on $R' = R \cup \{u\}$, $P' = P \cap \Nt(u)$ and $X' = X \cap \Nt(u)$. 
    First, $u \in P$, so by the induction hypothesis, $u$ is neighbor of all vertices in $R$, thus $R'$ is a clique of $G_t$.
    Second, $\underset{v \in R'}{\bigcap} \Nt(v) = \underset{v \in R}{\bigcap} \Nt(v) \cap \Nt(u) = (P \cup X) \cap \Nt(u) = (P \cap \Nt(u)) \cup (X \cap \Nt(u))$, so $\underset{v \in R'}{\bigcap} \Nt(v) = P' \cup X'$.
    Thus, the properties are indeed true at each recursive call.
  \end{description}

\end{proof}

% *********************
% Correct enumeration *
% *********************

\begin{lemma}
  \label{lemma:enum-timemax}
  For $t \in T$, the set $Cliques$ at Line~\ref{line:framework:forC} of Algorithm~\ref{algo:framework} contains exactly one pair $(C,\NC)$ per time-maximal clique $(C,[t,\EG{t}(C)])$ of L that begins at $t$.
\end{lemma}

\begin{proof}
  Let $(C,[t,\EG{t}(C)])$ be a time-maximal clique of $L$, that begins at $t$.
  We show by induction on $2 \leq k \leq |C|$ the following property
  $\P(k)$: \textit{``There is a call $\GCE{R,P,X,ForbidEdges,t}$ within Algorithm~\ref{algo:framework} such that $R$ satisfies $|R|=k$,  $R \subseteq C \subseteq R \cup P$ and $ForbidEdges$ contains no edge of $C$''}.
  \begin{description}
  \item[Initialization.]
    From Lemma~\ref{thm:timemax}~$(ii)$, there exists at least one link $(t,e,u,v)$ in the link stream that begins at $t$ with $u,v\in C$.
    Consider the first such link processed by the loop starting at Line~\ref{line:framework:fortuv}.
    Let us then consider the call to $\GCE$ of this iteration, at Line~\ref{line:framework:cliques}.
    In this call, $ForbidEdges$ does not yet contain any edge of $C$.
    Moreover, $R = \{u,v\}$ and $P = \Nt(u) \cap \Nt(v)$, and as all the elements of $C$ must be neighbors of $u$ and $v$, we have $R \subseteq C \subseteq R \cup P$.
    Thus, $\P(2)$ is true.
  \item[Induction.]
    Let $k$ be such that $2 \leq k \leq |C|-1$, we assume $\P(k)$ to be true, and show $\P(k+1)$.
    Let $\GCE{R,P,X,ForbidEdges,t}$ be a call of Algorithm~\ref{algo:framework} corresponding to $\P(k)$.
    From $\P(k)$, we know that $ForbidEdges$ contains no edge of $C$, therefore $C \cap Q = \emptyset$ (Line~\ref{line:staticenum:Q}) and the vertices of $C \setminus R$ (which is not empty) that may extend $R$ are in $P \setminus Q$.
    Without loss of generality, we consider the first such vertex $u \in C \setminus R$ processed by the loop starting at Line~\ref{line:staticenum:forup}.
    In the recursive call \GCE{$R\cup \{u\},P\cap \Nt(u),X\cap \Nt(u),ForbidEdges,t$}, $ForbidEdges$ is unchanged, so it contains no edge of $C$; moreover, by construction we have $R \cup \{u\} \subseteq C \subseteq P \cap \Nt(u)$, and $|R \cup \{u\}| = k+1$.
    Therefore, $\P(k+1)$ is true.
  \end{description}
  When $k=|C|$, we have shown that there is a call in Algorithm~\ref{algo:framework} that outputs $(C,\NC)$ at time $t$.
  We prove now that there cannot exist another call that outputs $(C,\NC)$ at time $t$.
  At Line~\ref{line:framework:addforbedge}, $\{u,v\}$ is added to $ForbidEdges$ after the first call involving $(t,e,u,v)$ with $u,v \in C$.
  Then, as  $Q$ (Line~\ref{line:staticenum:Q}) prevents the clique under construction from containing an edge of $ForbidEdges$, no call can enumerate $C$ in the following iterations.
  Finally, if we consider this iteration, the call to $\GCE$ cannot enumerate $(C,\NC)$ twice, because each recursive call in the loop starting at Line~\ref{line:staticenum:forup} enumerates different cliques, as $u$ is deleted from $P$ at Line~\ref{line:staticenum:Pu}.
\end{proof}

% ********************
% Correctness theorem *
% ********************

\begin{theorem}[Correctness]
  \label{thm:validity}
  Algorithm~\ref{algo:framework} lists once and only once each maximal clique of the input link stream $L$ and nothing else.
\end{theorem}

\begin{proof}
  We prove that for any $t \in T$, the related iteration of the loop starting at Line~\ref{line:framework:fort} enumerates all the maximal cliques that start at $t$, without duplication.
  Indeed, according to Lemma~\ref{lemma:enum-timemax}, the loop starting at Line~\ref{line:framework:forC}  iterates over all pairs $(C, \NC)$ such that $(C,[t,\EG{t}(C)])$ is a time-maximal clique starting at $t$ without duplication.
  Moreover, according to Lemma~\ref{lemma:outputGCE}, $\NC = P \cup X = \bigcap_{v \in C} \Nt(v)$
  so, according to the vertex-maximality test of Lemma~\ref{thm:vertexmax}, the time-maximal cliques output after the test at Line~\ref{line:framework:maxtest} are all of those which are vertex-maximal.
  Thus, at iteration $t$, the algorithm outputs once all maximal cliques that start at $t$ and nothing else.
\end{proof}

% *******************
% Correctness pivot *
% *******************

We proposed in Section~\ref{sec:pivot} to introduce a pivot in order to prune the recursive call tree of $\GCE$.
By cutting off branches corresponding to redundant computations, it improves the practical running times without changing the output of the algorithm.
We now prove that the algorithm with a pivot is correct.

\begin{theorem}[Correct pivoting]
  \label{thm:correctpivot}
  Algorithm~\ref{algo:framework} with pivoting is correct: adding a pivot as described in Section~\ref{sec:pivot} does not change the output of the algorithm.
\end{theorem}

\begin{proof}
  It is clear that adding a pivot to $\GCE$ as described in Section~\ref{sec:pivot} reduces the set of candidate nodes to add to the clique $R$ under construction.
  Thus, the pivot can only reduce the set $Cliques$ on Line~\ref{line:framework:forC}.
  So we can assert that all cliques output by the algorithm with a pivot are still maximal cliques of $L$ and that they are all distinct.
  
  It remains to show that each maximal clique is effectively output by the algorithm with a pivot.
  Let $(C,[t,\EG{t}(C)])$ be a maximal clique.
  Following the same scheme as the one that we used for the proof of Lemma~\ref{lemma:enum-timemax}, we show by induction on $2 \leq k \leq |C|$ the property
  $\P(k)$: \textit{``There is a call $\GCE{R,P,X,ForbidEdges,t}$ within Algorithm~\ref{algo:framework} with a pivot, such that $R$ satisfies $|R|=k$ with $R \subseteq C \subseteq R \cup P$ and $ForbidEdges$ contains no edge of $C$.''}
  Then, the call corresponding to $k=|C|$ would demonstrate the result.
  The initialization step is the same as in the proof of Lemma~\ref{lemma:enum-timemax}.
  Now assuming $\P(k)$ is true, given $k < |C|$ and $\GCE{R,P,X,ForbidEdges,t}$ the corresponding call, we only need to show that there is a vertex of $C \setminus R$ to extend $R$ into $C$ which is in $(P \setminus Del) \setminus Q$.
  Let us suppose that it is not the case, then $((P \setminus Del) \setminus Q) \cap C = \emptyset{}$ and since $ForbidEdges$ does not have any edge in $C$, $C \cap Q = \emptyset$, so we must have $C \subseteq R \cup Del$.
  Now, by definition of $Del$, we have  $R \cup Del \subseteq \Nt(p)$, so $C \subseteq \Nt(p)$.
  This implies two things: $p \notin C$, and $C \cup \{p\}$ is a clique of $G_t$.
  Moreover, $C \subseteq R \cup Del$ implies that $\EG{t}(C) = \EG{t}(C \cup \{p\})$ which contradicts the fact that $(C,[t, \EG{t}(C)])$ is a maximal clique.
  Hence, $\P(k+1)$ is true, and we deduce the expected result.
\end{proof}

\subsection{Complexity}

We compute here the complexity as a function of the input link stream parameters, and then as a function of the  features of the output. 
Throughout this section, we refer to the following characteristics of the link stream:

\begin{itemize}
\item[-] $|T|$: number of different time instants at which links begin or end, \emph{i.e.}:
  $$|T| = |\{t \in T \ | \ \exists (b,e,u,v) \in E, \ t=b \text{ or } t=e\}|$$
  
\item[-] $m=|E|$: the number of links;
  note that each link $(b,e,u,v)$ in the link stream corresponds to two times $b,e \in T$,
  therefore $|T| \leq 2m$;

\item[-] $d$: the maximal degree of a vertex in any static graph $G_t$;
\item[-] $\alpha_T$: the number of time-maximal cliques in $L$; 
  notice that not all time-maximal cliques are maximal and that
  each link induces a time-maximal clique, therefore $m \leq \alpha_T$;
\item[-] $\alpha$: the number of maximal cliques (\emph{i.e.} both time-maximal and vertex-maximal) in $L$; note that $\alpha \le \alpha_T$;
\item[-] $q$: the maximal number of vertices in a clique.
\end{itemize}

\subsubsection{Complexity as a function of the input link stream parameters}

To compute the time complexity of Algorithm~\ref{algo:framework}, we need the preliminary Lemma~\ref{lemma:alphaT}, which expresses it as a function of the number of time-maximal cliques $\alpha_T$.
This complexity implicitly takes into account the factors $|T|$ and $m$ of the link stream, since $\frac{1}{2}|T| \leq m \leq \alpha_T$.
Then, we will estimate an upper bound on $\alpha_T$ to deduce a general expression of this complexity.

% *************************
% Lemma complexity alphaT *
% *************************

\begin{lemma}
  \label{lemma:alphaT}
  The time complexity of Algorithm~\ref{algo:framework} is in $\mathcal{O}\left(d^2 \cdot \alpha_T\right)$.
\end{lemma}

\begin{proof}
  We recall that, by Lemma~\ref{lemma:enum-timemax}, for each $t \in T$, there is exactly one pair $(C,\NC)$ enumerated in the set $Cliques$ for each time-maximal clique of the link stream that begins at $t$.
  Thus, there are a total of $\alpha_T$ couples $(C,\NC)$ enumerated by the whole global loop 
  of Line~\ref{line:framework:fort}.
  
  First, let us analyze the complexity of a recursive call to \GCE.
  Let us consider each operation associated to such a recursive call and show that they never exceed~$\mathcal{O}\left(d^2\right)$:
  \begin{itemize}
  \item[-] computing its arguments (at Line~\ref{line:framework:cliques} or~\ref{line:staticenum:RC}) is done by intersection of sets of size at most $d$, so it is in~$\mathcal{O}\left(d\right)$;
  \item[-] the sets in the couple $ (R , P \cup X )$ output at Line~\ref{line:staticenum:output} are of size at most $d$, so the output is in~$\mathcal{O}\left(d\right)$;
  \item[-] computation of $Q$ at Line~\ref{line:staticenum:Q} corresponds, for each element $u$ of $P$, to the intersection between $R$ and the neighbors of $u$ in $ForbidEdges$.
    Each of the three sets mentioned are of size at most $d$, so this calculation is done in $\mathcal{O}\left(d^2\right)$;
  \item[-] each iteration of the loop starting on Line~\ref{line:staticenum:forup} is constituted of the computation of the arguments of another recursive call (Line~\ref{line:staticenum:RC}) counted in this one's complexity, and constant time operations (Lines~\ref{line:staticenum:Pu} and~\ref{line:staticenum:Xu}).
  \end{itemize}
  Finally, each pair $(C,\NC)$ is the output of a recursive call to \GCE, and each recursive call outputs one such pair.
  Since there are $\alpha_T$ recursive calls following the initial remark, the complexity of the corresponding operations is therefore in $\mathcal{O}\left(d^2\cdot \alpha_t\right)$.
  
  Second, the processing of each pair $(C,\NC)$ by the maximality test of Lines~\ref{line:framework:maxtest} and~\ref{line:framework:output} is done in $\mathcal{O}\left(d^2\right)$ and thus the complexity of all maximality tests done by the entire algorithm is in $\mathcal{O}\left(d^2 \cdot \alpha_T\right)$.
  Indeed, each test requires computing the final times $\EG{t}(C)$, and $\EG{t}(C \cup \{u\})$ for $u$ in $\NC$.
  Notice that the link stream data structure allows direct access to the end time of each link, so during successive calls to $\GCE$, it is possible to maintain the end time of the current clique $R$ in $\mathcal{O}\left(d\right)$.
  This is because when we add a vertex $u$ to $R$ at Line~\ref{line:staticenum:RC} of Algorithm~\ref{algo:framework}, then $\EG{t}(R \cup \{u\}) = \min \left\lbrace \EG{t}(R), \min_{v \in R} \{ \EG{t}(u,v) \} \right\rbrace  $, therefore this operation can be carried out in $\mathcal{O}\left(|R|\right) \subseteq \mathcal{O}\left(d\right)$. 
  Then, access to $\EG{t}(C)$ is possible in $\mathcal{O}\left(d\right)$ operations.
  Now, $\EG{t}(C \cup \{u\})$ is obtained by adding one vertex to $C$, so with the same procedure in $\mathcal{O}\left(|C|\right) \subseteq \mathcal{O}\left(d\right)$.
  Since $\NC$ is at most of size $d$, it is computed in $\mathcal{O}\left(d\right)$. 
  Moreover, the output of Line~\ref{line:framework:output} has also a complexity in $\mathcal{O}\left(d\right)$.
  Thus, each pair $(C,\NC)$ is processed in $\mathcal{O}\left(d^2\right)$, so the total complexity of the maximality tests made by the entire algorithm is in $\mathcal{O}\left(d^2 \cdot \alpha_T\right)$.

  Finally, the overall complexity, which is the sum of the two above, is in $\mathcal{O}\left(d^2 \cdot \alpha_T\right)$.
\end{proof}

% *************************
% Theorem Time complexity *
% *************************

We can now show the overall time complexity of our algorithm.

\begin{theorem}[Time complexity]
  \label{thm:timecomplexity}
  The time complexity of Algorithm~\ref{algo:framework} is in $\mathcal{O}\left(m \cdot 3^{d/3} \cdot 2^q \cdot d^2\right)$.
\end{theorem}

\begin{proof}
  Moon and Moser showed in~\cite{moon1965cliques} that the number of maximal cliques in a graph with $n$ vertices is in $\mathcal{O}\left(3^{n/3}\right)$. 
  Now, a call to \GCE in Line~\ref{line:framework:cliques} outputs pairs $(C,\NC)$ where each $C$ is a clique (different from the others) of the sub-graph composed of the vertices of $\Nt(u) \cap \Nt(v) \cup \{u,v\}$.
  As there are at most $d+2$ vertices in this sub-graph, the number of pairs $(C,\NC)$ output by this call such that $C$ is a maximal clique of $G_t$, is in $\mathcal{O}\left(3^{d/3}\right)$.
  Now, each of these maximal cliques  is of size at most $q$ by definition, so it contains at most $2^q$ different sub-cliques.
  Therefore, the total number of pairs $(C,\NC)$ enumerated by such a call is in $\mathcal{O}\left(2^q \cdot 3^{d/3}\right)$. 
  Finally, such a call is made for each link of the link stream, so the total number of pairs $(C,\NC)$ enumerated during Algorithm~\ref{algo:framework} is in $\mathcal{O}\left(m \cdot 2^q \cdot 3^{d/3}\right)$.
  In other words, $\alpha_T$ is in $\mathcal{O}\left(m \cdot 2^q \cdot 3^{d/3}\right)$.
  Finally, from Lemma~\ref{lemma:alphaT}, the time complexity of the algorithm is in $\mathcal{O}\left(m \cdot 3^{d/3} \cdot 2^q \cdot d^2\right)$.
\end{proof}

By contrast, the complexity of the algorithm in Viard~\etal~\cite{viard2018} is in $\mathcal{O}\left(2^n \cdot n^2 \cdot m^2 \cdot (n + \log(m))\right)$ with $n = |V|$; in Himmel~\etal~\cite{himmel2017adapting} it is in $\mathcal{O}\left(m \cdot 3^{c/3} \cdot 2^c \cdot n \cdot |T|\right)$ with $c$ the maximal degeneracy of a graph $G_t$ (notice that $q-1 \leq c \leq d$); in Bentert~\etal~\cite{bentert2019}, it is in $\mathcal{O}\left(2^c \cdot \min(m^2,|T|^2) \cdot n^4\right)$.
While not directly comparable to ours, these three complexities are all products of at least two of the parameters $n$, $m$ and $|T|$ of the link stream, while there is only one single factor $m$ among those in ours (we recall that $|T| \leq 2m$).
Note that all these complexities also depend on $c$, $q$ or $d$, but these parameters are known to vary in much smaller ranges than $n$, $m$ and $|T|$ when considering real data.
Thus, this observation suggests that our algorithm might be more able to scale up to larger link streams, which we confirm experimentally in Section~\ref{sec:experiments}.

Regarding the complexity of the algorithm with a pivot, notice that in the worst case, each time-maximal clique is a maximal clique of the link stream.
In that case, the pivot cannot prune any recursive call.
Thus, the worst-case complexity is the same as the one of the version without a pivot.
Nevertheless, we will study in more details its impact in Section~\ref{sec:output-sensitive-complexity} in which we study the output-sensitive complexity.

% ***************************
% Theorem Memory complexity *
% ***************************

Finally, we show in Theorem~\ref{thm:memorycomplexity} that the memory requirement of our algorithm is close to the size $m$ of the link stream.
It is possible to show that the algorithms proposed in~\cite{himmel2017adapting} and~\cite{bentert2019}  have the same memory complexity as ours, while the one in~\cite{viard2018} is exponential.

\begin{theorem}[Memory complexity]
  \label{thm:memorycomplexity}
  The spatial complexity of Algorithm~\ref{algo:framework} is in  $\mathcal{O}\left(m + q \cdot d\right)$.
\end{theorem}

\begin{proof}
  In practice, we store the $m$ links of the link stream.
  We also need to store $ForbidEdges$ which is at most of size \(m \).
  Notice that it is possible to perform the
  maximality test of Line~\ref{line:framework:maxtest} inside the function
  $\GCE$,
  which eliminates the need to store more than one clique at any given time,
  it can thus be done using a memory space in $\mathcal{O}\left(q\right)\subseteq \mathcal{O}\left(m\right)$.
  Finally, we have to compute the memory cost of the calls to $\GCE$.
  Each call at Line~\ref{line:framework:cliques} generates a stack of recursive calls of size at most $q$ and for each of them, we store $R$, $P$, $X$ and $Q$, each of these sets being of size at most $d$.
  The stack therefore demands a space complexity in $\mathcal{O}\left(q \cdot d\right)$.
  All these data structures add up to a memory complexity in $\mathcal{O}\left(m + q \cdot d\right)$.
\end{proof}

\subsubsection{Output-sensitive complexity}
\label{sec:output-sensitive-complexity}

We formulate here the time complexity as a function of $\alpha$, the number of maximal cliques output by the algorithm, and $q$ the maximum number of vertices in a clique.
To do this, we consider the search trees of recursive calls to \GCE.
The internal nodes of these trees correspond to the calls for which the loop on Line~\ref{line:staticenum:forup} is not empty (i.e. which generates other child calls), while the leaves correspond to the calls that do not generate any subsequent call.
Note that all leaves correspond to time maximal cliques.

Inspired by the work of Conte~\etal~\cite{conte2022}, in what follows, we focus on the leaves of these search trees, which we separate into two categories: those which output a pair $(C,\NC)$ that corresponds to a maximal clique of the link stream, and those whose output pair does not correspond to a maximal clique. 
The latter create unnecessary computations because they do not contribute to the enumeration. 
An optimal pivot strategy would cut off the branches of the search forest so that each leaf always outputs a maximal clique.
We denote $ \l $ the total number of leaves of the search forest, $\lm$ the number of leaves that correspond to maximal cliques, and $\lnm$ the ones which are not, so that $ \l = \lm+\lnm$.
We are then interested in the ratio of ``good'' leaves: 
$$r = \dfrac{\lm}{\l} =\dfrac{\lm}{\lm+\lnm}.$$ 
This ratio can be computed either for the algorithm without a pivot or the one with a pivot.
In the first case, it quantifies the maximal possible efficiency of the pivot:
if $r$ is lower than 1 this means that there are recursive calls which are not necessary
and might be pruned by a pivot.
Comparing the ratios obtained by the algorithms with and without the pivot shows to what extent the pivoting strategy has
successfully pruned unnecessary calls.

Using this ratio allows us to describe the time complexity of Algorithm~\ref{algo:framework} as a function of the~output:

\begin{theorem}[Output-sensitive complexity]
  \label{thm:output-sensitive-complexity}
  With the above definition of $r$, we have $1 \leq \frac{1}{r} \leq 2^q$, and the output-sensitive complexity of Algorithm~\ref{algo:framework} is in $\mathcal{O}\left(\frac{1}{r} \cdot d^2 \cdot q \cdot \alpha\right)$.
\end{theorem}

\begin{proof}
  Let us show first that $1 \leq \frac{1}{r} \leq 2^q$.
  For this, we introduce $c_{max}$: the number of pairs $(C,\NC)$ enumerated by $\GCE$ such that $C$ is a maximal clique of the associated graph $G_t$.
  These pairs can only be enumerated by calls corresponding to search tree leaves, because if $C$ is a maximal clique of $G_t$, then it cannot be expanded by a recursive call.
  Moreover, since it cannot be expanded, its associated time-maximal clique is always vertex-maximal.
  Thus, $c_{max}$ is such that $c_{max} \leq \lm$.
  Besides, each maximal clique of a graph $G_t$ contains at most $2^q$ sub-cliques, so there are at most $2^q \cdot c_{max}$ pairs $(C,\NC)$ listed in total.
  Now, there is at least one pair listed per leaf.
  Therefore, $\l \leq 2^q \cdot c_{max}$, so $\frac{1}{r} \leq \frac{2^q \cdot c_{max}}{\lm} \leq 2^q$.

  Let us now prove the expression for the time complexity.
  By definition of $q$, we know that the depth of a search tree is at most $q$.
  Then, there are at most $q \cdot \l$ nodes in the whole search forest. 
  We have seen in the proof of Lemma~\ref{lemma:alphaT} that there is exactly one node per time-maximal clique.
  So, there are $\alpha_T$  nodes, which implies that $\alpha_T \leq q \cdot \l$.
  So, using the expression of the time complexity in Lemma~\ref{lemma:alphaT}, we have that the complexity of Algorithm~\ref{algo:framework} is in $\mathcal{O}\left(d^2 \cdot q \cdot \l\right)$. 
  As $\l = \frac{1}{r} \cdot \lm$ and $\lm \leq \alpha$, the complexity is in $\mathcal{O}\left(d^2 \cdot q \cdot \frac{1}{r} \cdot \alpha\right)$.
\end{proof}

Note that the bounds given by Theorem~\ref{thm:output-sensitive-complexity} are almost tight. % narrow.
First, $\frac{1}{r}$ can be equal to $1$ for example when all the time-maximal cliques are also vertex-maximal, as in this way $\lnm=0$.
This is the case for instance when all the links end times are distinct.
However, $\frac{1}{r}$ can also be exponential in $q$ as suggested by its upper bound.
To illustrate this, consider the example of a link stream
that is equivalent to a graph because all its links have the same existence interval.
Suppose that this link stream forms a clique.
Then, it contains only one maximal clique, so that $\frac{1}{r} = \l$.
With the \BK procedure on graphs (without pivot) described in Section~\ref{sec:BKStatic}, it is possible to show that the call tree has $\l{} = 2^{q-1}-1$ leaves of size $\geq 2$.
Thus, Algorithm~\ref{algo:framework} on this link stream can reach~$\frac{1}{r} = 2^{q-1}-1$.

Nevertheless, we will see in Section~\ref{sec:experiments} that $\frac{1}{r}$ is experimentally small with a pivot: it is lower than 2 in all experiments except one.
This means that it does not have an exponential behavior on the real world link streams that we study, by contrast with its theoretical upper bound.
According to this observation and the computed complexity, we can state that the running time of our algorithm with a pivot is close to the best that we can expect.
Indeed, since the running time must process all vertices of all maximal cliques (of which there are $\alpha{}$) and the size of the largest ones is $q$,
it is expected to have a factor $q \cdot \alpha$ in the complexity.
Our algorithm only adds a multiplicative factor $\frac{1}{r} \cdot d^2$ to this, thus explaining its good performances. 

%% file: experiments.tex
\section{Experimental evaluation}
\label{sec:experiments}

In this section, we compare the implementations of the algorithm that we have presented in Section~\ref{sec:algorithm} to the best ones provided in the literature.
These are those of Viard~\etal~\cite{viard2018}\,\footnote{\url{https://bitbucket.org/tiph_viard/cliques}},
Himmel~\etal~\cite{himmel2017adapting}\,\footnote{\url{https://fpt.akt.tu-berlin.de/temporalcliques/}},
and Bentert~\etal~\cite{bentert2019}\,\footnote{\url{https://fpt.akt.tu-berlin.de/temporalkplex/}}.
We do not compare our algorithm to the one of Banerjee and Pal~\cite{banerjee2022efficient} as they study a generalization of the maximal clique enumeration and their implementation for the special case corresponding to our specific problem is not more efficient than the other algorithms mentioned.
We do not compare either with the work of Molter~\etal~\cite{molter2021isolation}, as the authors explain that their algorithm is slower in all cases than that of Bentert~\etal

Then, we show that our algorithm scales to massive real-world link streams of more than 100 million links, we study the impact of the pivot strategy on the computational time, and finally we present results of a parallel implementation.

\subsection{Experimental setup}

\paragraph{Hardware.}

We carried out the experiments on a machine equipped with 2 processors Intel Xeon Silver 4216 with 32 cores each and 384 GB of RAM.

\paragraph{Implementations.}

We have made two implementations of Algorithm~\ref{algo:framework} which are available online\,\footnote{\url{https://gitlab.lip6.fr/baudin/maxcliques-linkstream}}: one in \texttt{Python} and another in \texttt{C++}.
We use the \texttt{Python} implementation for comparison with the state-of-the-art methods which are coded in \texttt{Python}, and the \texttt{C++} one to scale up to more massive link streams.
The \texttt{C++} implementation is inspired by the efficient implementation of the BK algorithm by Eppstein~\etal~\cite{eppstein2010} that enumerates maximal cliques on graphs.
Since using a pivot reduces the running time, we use by default the algorithm with a pivot (as detailed in Section~\ref{sec:pivot}),
except in Section~\ref{sec:pivoting}, in which we analyze the efficiency of the pivot strategy.

\paragraph{Data and pre-processing.}

Following the work of Viard~\etal in~\cite{viard2018}, we use datasets corresponding to {\em instantaneous} link streams in which all links have a duration equal to 0:
we associate a duration $\Delta$ to each link $(t, u,v)$ by replacing it by $(t,t+\Delta,u,v)$.
Note that this may create overlapping links $(b,e,u,v)$ and $(b',e',u,v)$ with $[b,e]\cap [b',e'] \not= \emptyset{}$;
we resolve this by iteratively replacing such pairs of links by $(\min(b,b'), \max(e,e'), u,v)$.
Notice that this transformation implies that the number of links may decrease as $\Delta$ increases.
This allows us to study the behavior of our algorithm with varied input, while still allowing to compare it to the ones that study $\Delta$-cliques in instantaneous link streams, since Viard~\etal~\cite{viard2018} showed that both problems are equivalent.
This pre-processing is performed before the main algorithm, and is therefore not taken into account in the computation times.
Nevertheless, its running time is linear in the number of links in the instantaneous link streams and is negligible in practice when compared to the enumeration time of maximal cliques.

We use two different families of datasets. 
The first one, whose main characteristics are detailed in Table~\ref{tab:data-bentert}, corresponds to the datasets used in Bentert~\etal~\cite{bentert2019} to compare their own algorithm to the literature. 
In this case, $\Delta$ values are chosen identical to those used in~\cite{bentert2019}, for comparison purposes.
The second one corresponds to a set of more massive datasets, whose characteristics are given in Table~\ref{tab:data-bigLS}.
We use these datasets to show that our algorithm scales to link streams up to several tens of millions of links.
In this family, $\Delta$ values are chosen as functions of $ \Theta $, the total duration of the link stream: either 0 (instantaneous), $ \Theta / 10000 $, or $ \Theta / 100 $.
The results of the corresponding experiments are detailed respectively in Section~\ref{sec:stateofart} and Section~\ref{sec:bigLS}.

These two tables illustrate the effect of increasing $\Delta$: the number of links $m$ decreases (as explained above), while the maximum degree $d$ and maximum clique size $q$ increase. 
Indeed, increasing link duration naturally makes instantaneous graphs denser overall, which leads to an increase of $d$ and $q$.
However, the effect on the number of maximal cliques $\alpha$ is not the same in all cases.
In some cases, $\alpha$ increases, probably because the increased density in instantaneous graphs induces more cliques.
In other cases $\alpha$ decreases, possibly because cliques involving the same vertices at different times are merged.

\begin{table}[!hbt]
  \centering
  \setlength{\tabcolsep}{4pt} 
  {\scriptsize \input{Results/data-bentert}}
  \caption{
    Characteristics of the link stream datasets investigated by Bentert~\etal~\cite{bentert2019}.
    $\boldsymbol{\Delta}$ (in seconds) corresponds to the duration added to each link of the stream,
    $\boldsymbol{m}$ is the number of links, $\boldsymbol{d}$ the maximal degree, $\boldsymbol{\alpha}$ the number of maximal cliques, and $\boldsymbol{q}$ the maximal number of nodes of a clique.
  }
  \label{tab:data-bentert}
\end{table}

\begin{table}[!hbt]
  \centering
  {\footnotesize \input{Results/data-bigLS}}
  \caption{
    Characteristics of large link stream datasets used to investigate the scaling properties of the algorithm.
    $\Delta$ values (in seconds) are set to 0, $\Theta/10000$ and $\Theta/100$, where $\Theta$ is the total duration of the link stream. 
    A hyphen (``-'') means that none of the implementations available allows enumerating all maximal cliques.
  }
  \label{tab:data-bigLS}
\end{table}

\subsection{Comparison to the state of the art}
\label{sec:stateofart}

In Table~\ref{tab:time-all-bentert} we present the computation times (in seconds) for the datasets of Table~\ref{tab:data-bentert}, which are  those studied by
Bentert~\etal~\cite{bentert2019}.
As the state-of-the-art algorithms are implemented in \texttt{Python},
we first compare them with our \texttt{Python} implementation to evaluate the
improvement brought by our algorithm.
We observe that our algorithm is the fastest in all tested cases.
Among the three state-of-the-art algorithms, \textbf{HMNS}~\cite{himmel2017adapting} is clearly the least competitive, while \textbf{BHM+}~\cite{bentert2019} and \textbf{VML}~\cite{viard2018} algorithms yield comparable results in terms of running times: depending on the link stream, one or the other is faster, while remaining in the same order of magnitude.
Besides, we also observe that our implementation in \texttt{C++} is very efficient on these datasets: in most cases, the computation time is less than 1 second, and never exceeds 3 seconds, even on the \textit{flights} dataset, for which it is typically 1,000 times faster than the state of the art, and 30 times faster than our {\tt Python} implementation.

To study more precisely the difference of performance, we visualize the computation times of the different methods on all link streams in Figure~\ref{fig:best-soa} (left panel).
The different link streams  are ordered on the X-axis by increasing number of links $m$.
Note that the scales are logarithmic.
There is one experiment per value of $\Delta$ for each dataset, and three values are displayed, corresponding to the running times of the fastest state-of-the-art method, our \texttt{Python} implementation, and our \texttt{C++} implementation.
The three horizontal lines at the top of the figure correspond to the runs for which the computation is interrupted, \textit{i.e.} exceeds 24 hours or 380 GB of RAM.
We also placed a vertical line to show the limit above which the state-of-the-art algorithms fail to provide a result.
We observe three distinct layers of points, which illustrates that the \texttt{C++} implementation is more efficient than the \texttt{Python} implementation, itself more efficient than the state of the art.
Finally, we notice a trend to have slightly higher gain with larger link streams.

Then, to better evaluate the gain, we display the speed-up factor achieved by our two implementations in the right panel of Figure~\ref{fig:best-soa}, compared to the most efficient algorithm of the state of the art.
We display a point for each dataset where at least one state-of-the-art algorithm is able to provide a result in less than 24 hours and 380 GB of RAM.
The speed-up factor is defined as the state-of-the-art execution time divided by the time of our {\tt Python} and {\tt C++} implementations.
The  position of the line $y=1$ at the bottom shows that this factor is always larger than 1.
Thus, both {\tt Python} and {\tt C++} implementations are more efficient than the other implementations available, in all the experiments performed.
In most experiments, the speed-up factor is larger than 10, even for the \texttt{Python} implementation.
Concerning the \texttt{C++} implementation the speed-up factor is always larger than 10 except for the 3 smallest link streams, and it can reach up to $10^4$ for larger link streams.

\begin{table}[!hbt]
  \centering
  \setlength{\tabcolsep}{0.5pt} 
  {\scriptsize  \input{Results/time-all-bentert}}
  \caption{
    Comparison of the computation times (in seconds) of our implementations (\textbf{C++} and \textbf{Python}) to the state-of-the-art implementations: \textbf{BHM+}~\cite{bentert2019}, \textbf{VML}~\cite{viard2018} and \textbf{HMNS}~\cite{himmel2017adapting} on the link streams listed in~\cite{bentert2019} described in Table~\ref{tab:data-bentert}.
    A ``-'' symbol means that the computation time exceeds 24 hours, and a ``$\times$'' symbol means that the memory needed for the computation exceeds~380 GB.
  }
  \label{tab:time-all-bentert}
\end{table}

\begin{figure}[!hbt]
  \centering
  \includegraphics[width=0.48\linewidth]{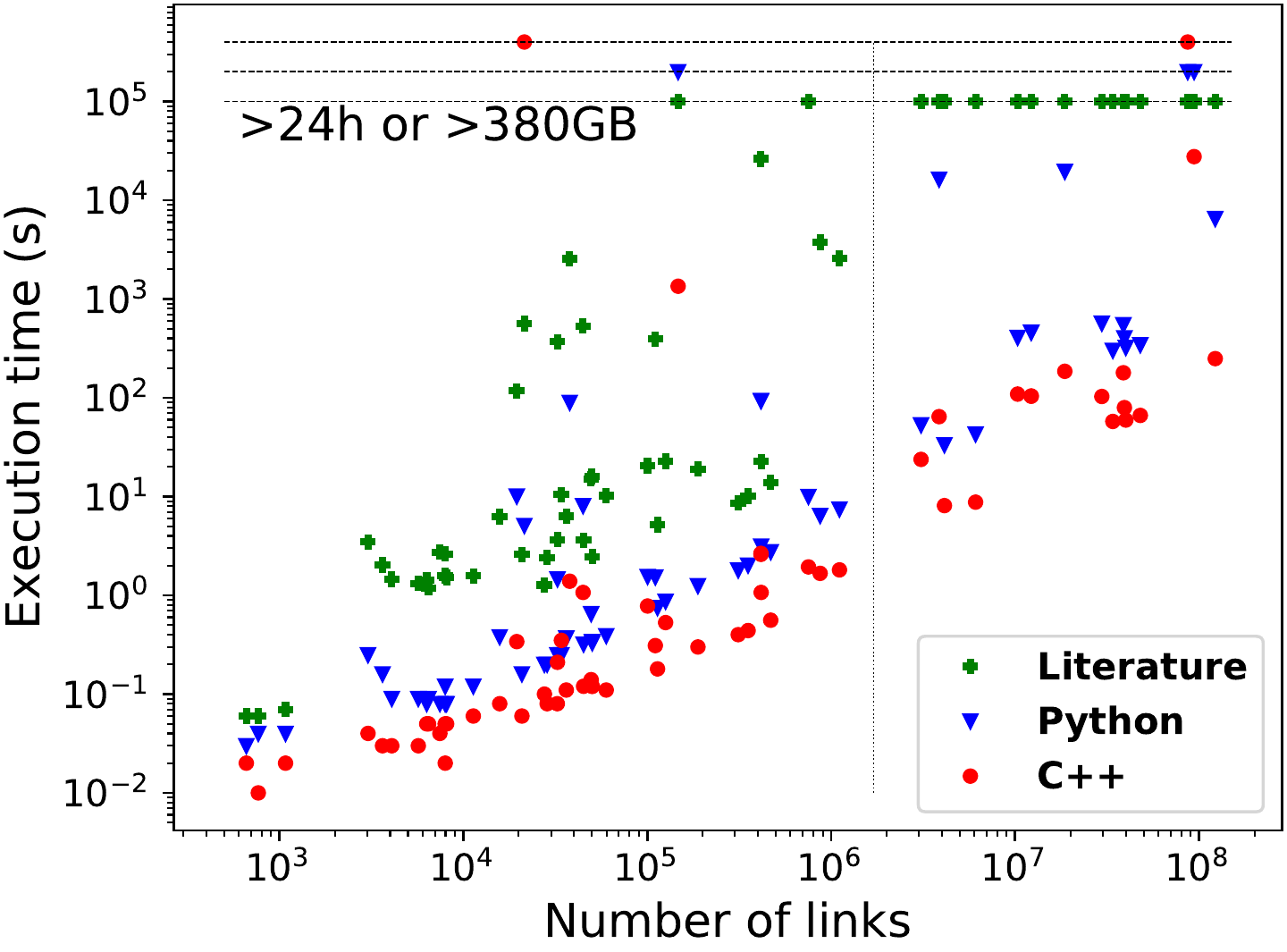}
  \hfill
  \includegraphics[width=0.48\linewidth]{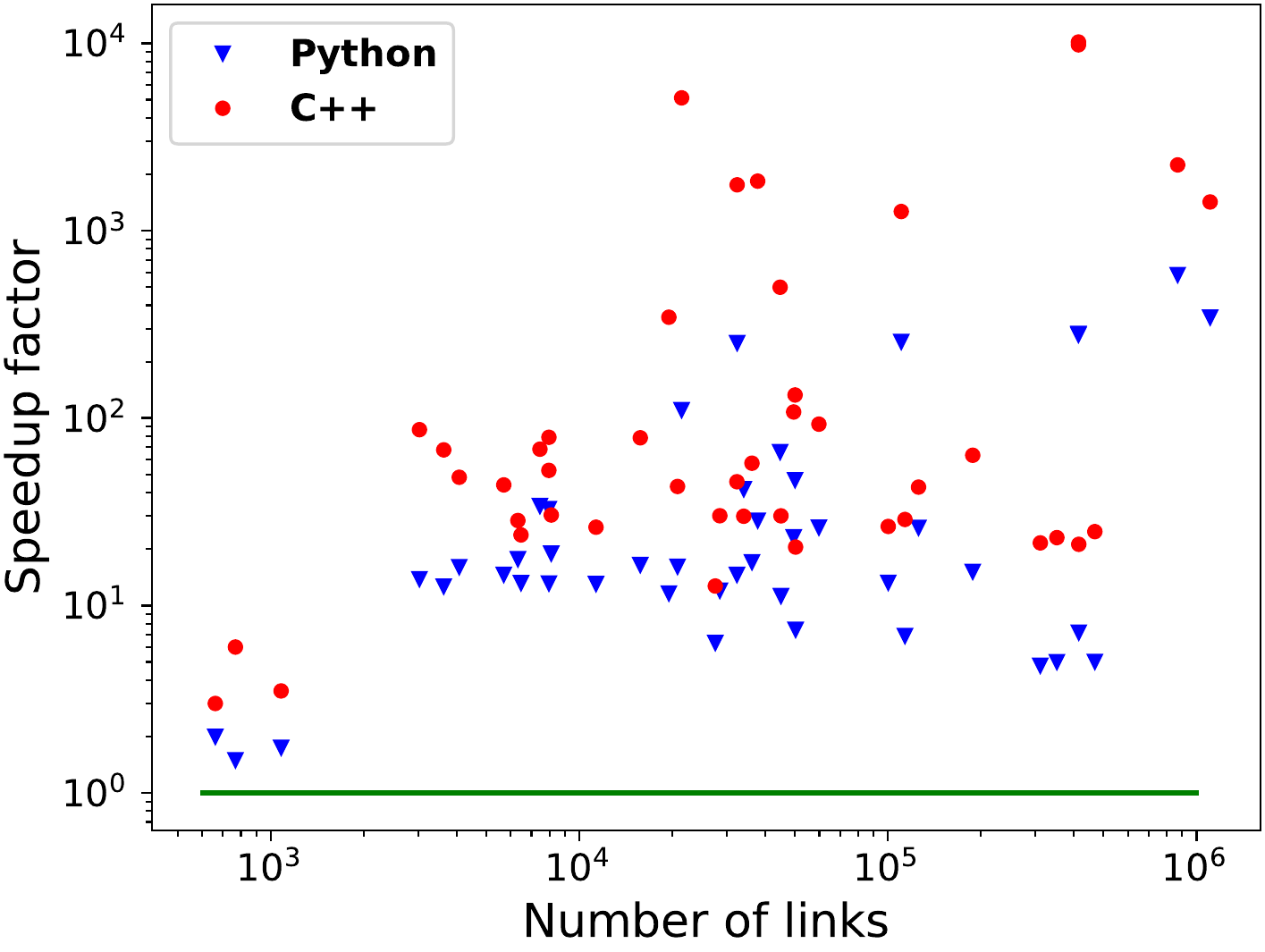} 
  \caption{
    \textbf{Left}: Summary of the computation times of maximal clique enumerations as a function of the number $m$ of links for all link stream datasets in Tables~\ref{tab:time-all-bentert} and~\ref{tab:time-all-bigLS}.
    The three lines at the top represent enumerations that are interrupted because they exceed 24 hours or 380 GB of RAM.
    \textbf{Right}:
    Speed-up factor of our implementations with respect to the fastest state-of-the-art method, as a function of the number $m$ of links.
    There is one point per dataset where at least one state-of-the-art algorithm finishes in less than 24 hours and using less than 380 GB RAM.  
  }
  \label{fig:best-soa}
\end{figure}

\subsection{Scaling up to massive real-world link streams}
\label{sec:bigLS}

In the previous section, we have seen that our algorithm is more efficient
than those of the state of the art.
It is thus relevant to see to what extent it can be used on larger datasets.
Table~\ref{tab:time-all-bigLS} presents the running times of the maximal clique enumerations performed on the massive link streams described in Table~\ref{tab:data-bigLS}.
In this table, a ``-'' symbol means that the enumeration does not finish within 24 hours for the algorithm under examination and a ``$\times$'' symbol indicates that it requires more than 380 GB of RAM.

We observe that the state-of-the-art implementations do not allow enumerating the maximal cliques on these datasets, except for \textbf{VML}~\cite{viard2018} on the dataset \textit{stackexchange} with $\Delta = 0s$ and $\Delta = 23,961s$, which is much less efficient than both our implementations.
In all other cases, the experiments are interrupted either for time or for memory consumption reasons.
The \texttt{Python} implementation of our algorithm allows enumerating maximal cliques for all experiments except for~3.
The \texttt{C++} implementation allows scaling the computation further, as it completes the enumeration for all experiments except for~1: the \textit{soc-bitcoin} link stream with link duration $\Delta = 1,565,366s$. 
We observe a significant gain of the \texttt{C++} implementation compared to the \texttt{Python} implementation.

\begin{table}[!hbt]
  \centering
  {\footnotesize \input{Results/time-all-bigLS}}
  \caption{
    Comparison of the computation times (in seconds) of our implementations (\textbf{C++} and \textbf{Python}) to the state-of-the-art implementations: \textbf{BHM+}~\cite{bentert2019}, \textbf{VML}~\cite{viard2018} and \textbf{HMNS}~\cite{himmel2017adapting} on the massive link streams described in Table~\ref{tab:data-bigLS}.
    A ``-'' symbol means that the computation time exceeds 24 hours, and a ``$\times$'' symbol means that the memory needed for the computation exceeds 380 GB.
  }  
  \label{tab:time-all-bigLS}
\end{table}

\subsection{Efficiency of the pivot}
\label{sec:pivoting}

In this section, we analyze and comment on the gain of computation time achieved by using a pivot, which was presented in Section~\ref{sec:pivot}.
We recall that the complexity of Algorithm~\ref{algo:framework}, as given by Theorem~\ref{thm:output-sensitive-complexity}, is in  $\mathcal{O}\left(\frac{1}{r} \cdot d^2 \cdot q \cdot \alpha\right)$, where $d$ is the maximal degree in an instantaneous graph $G_t$, $q$ is the maximal size of a clique, $\alpha$ is the number of maximal cliques and $r$ is the ratio of leaves that corresponds to maximal cliques in the tree of calls of $\GCE$.
$r$ can be computed for the enumeration performed without or with the pivot.
In the first case, $r$  quantifies the potential gain that the pivot can bring;
in the second case, it quantifies how efficient the pivot actually is: the closer $r$ gets to 1, the fewer useless branches remain in the tree.

In Table~\ref{tab:pivoting} we give the computation times of the \texttt{C++} implementation with and without a pivot for the massive link streams of Table~\ref{tab:data-bigLS}.
In both cases we report the enumeration times $t$, the ratio $r$ mentioned above and the factor $1/r$, as it appears in the complexity expression of Theorem~\ref{thm:output-sensitive-complexity}.

We first observe that the pivot allows to achieve the enumeration of maximal cliques faster.
The speed-up factor may vary a lot from a dataset to another and in 5 cases, the pivot version terminates within our time limit while the version without pivot does not.
The ratio $r$ is larger than $0.9$ in 18 experiments out of 21 with the pivot, while it is lower in most experiments without it.
This indicates that the pivot allows to cut off almost all unnecessary branches of recursive calls,  except for the \textit{copresence-Thiers} dataset.
Notice that this dataset with a link duration $\Delta = 0s$ clearly stands out, as there are only $8.7\%$ of the leaves of the call trees of $\GCE$ which correspond to maximal cliques for the enumeration with the pivot.

We can also see that the factor $1/r$, which can be very large in theory since it is in $\mathcal{O}\left(2^q\right)$  according to Theorem~\ref{thm:output-sensitive-complexity},
remains relatively small in  the experiments.
This factor never exceeds 2 with the pivot, except in the case of \textit{copresence-Thiers} with $\Delta = 0s$ mentioned above. 
Thus, from this observation, the complexity with the pivot can often be considered in practice as in $\mathcal{O}\left(d^2 \cdot q \cdot \alpha\right)$, which is within a $d^2$ factor of the output size (in $\mathcal{O}\left(q \cdot \alpha\right)$).

Finally, even with a pivot, the computation does not succeed for the largest value of $\Delta$ on the \textit{soc-bitcoin} dataset within the boundaries of the experimental procedure. 
There are other larger datasets on which no method is able to produce a result, as for instance the link stream \textit{flickr}~\cite{mislove2008growth} which has very dense instantaneous graphs. 
Based on observations not reported here, we suggest that this is not due to the size of the output, but rather to the fact that the pivot does not succeed in pruning a sufficient number of branches.

\begin{table}[!hbt]
  \centering
  {\footnotesize \input{Results/pivoting}}
  \caption{
    Comparison of the computation times, in seconds, of the \texttt{C++} implementation using a pivot to the one without a pivot.
    The factor $r$, defined in Section~\ref{sec:output-sensitive-complexity}, is equal to the ratio of leaves in the call trees of \protect\GCE that correspond to a maximal clique of the link stream. 
  }
  \label{tab:pivoting}
\end{table}

\subsection{Parallel experiments}
\label{sec:parallelization}

In this section, we study a parallel version of the code to evaluate the speed-up that it brings.
Algorithm~\ref{algo:framework} is indeed easily parallelizable, as the iterations of the loop on $T$ at Line~\ref{line:framework:fort} are independent of each other.
Our procedure consists in splitting the total time interval of duration $ \Theta $ of the link stream into $n_{th}$ sub-intervals, where $n_{th}$ is the number of threads on which we perform the parallelization.
For each sub-interval in parallel, the corresponding thread enumerates all maximal cliques that start during this interval, following the loop of Line~\ref{line:framework:fort}.
We choose to split the total time interval of the link stream in such a way that approximately the same number of links starts within each sub-interval.
Thus, each thread processes approximately $\frac{m}{n_{th}}$ links.
According to the expression established by Theorem~\ref{thm:timecomplexity}, the sequential complexity is in $\mathcal{O}\left(m \cdot 3^{d/3} \cdot 2^q \cdot d^2\right)$, so by dividing the number of links by $n_{th}$, the theoretical complexity is now in $\mathcal{O}\left(\frac{1}{n_{th}} \cdot m \cdot 3^{d/3} \cdot 2^q \cdot d^2\right)$ per thread.

Figure~\ref{fig:parallel} illustrates the results of the parallelization process on the massive link stream datasets detailed in Table~\ref{tab:data-bigLS}.
It reports the execution time for the enumeration of maximal cliques as a function of the number of threads used.
On the left are the link streams for which parallelization offers an interesting reduction in computation times, while we show on the right the link streams for which parallelization does not yield any significant improvement.
We observe different behaviors depending on the link stream under study: for example, the parallelization of \emph{wikipedia} with $\Delta= \Theta~/~100$ leads to a division by up to three of the enumeration time and even more for \emph{copresence-Thiers} with $\Delta=0s$, while the same process on link streams on the right brings very little gains.
Also, we observe that for almost all datasets, the computation times do not decrease significantly when using more than 4 threads, by contrast with the theoretical expression of the complexity established above.
This low speed-up can be explained by the distribution of the link stream density through time.
Indeed, the density of the link stream is the maximal density of $G_t$ at each instant $t$, which is not affected by the parallelization process.
Even a short sub-interval of time can be associated to a dense link stream and thus generate many maximal cliques.
As the overall computation time is set by the sub-interval that takes the longest time to compute, a dense sub-stream creates a bottleneck and entails a low speed-up.

\begin{figure}[!ht]
  \centering

  \includegraphics[width=0.48\linewidth]{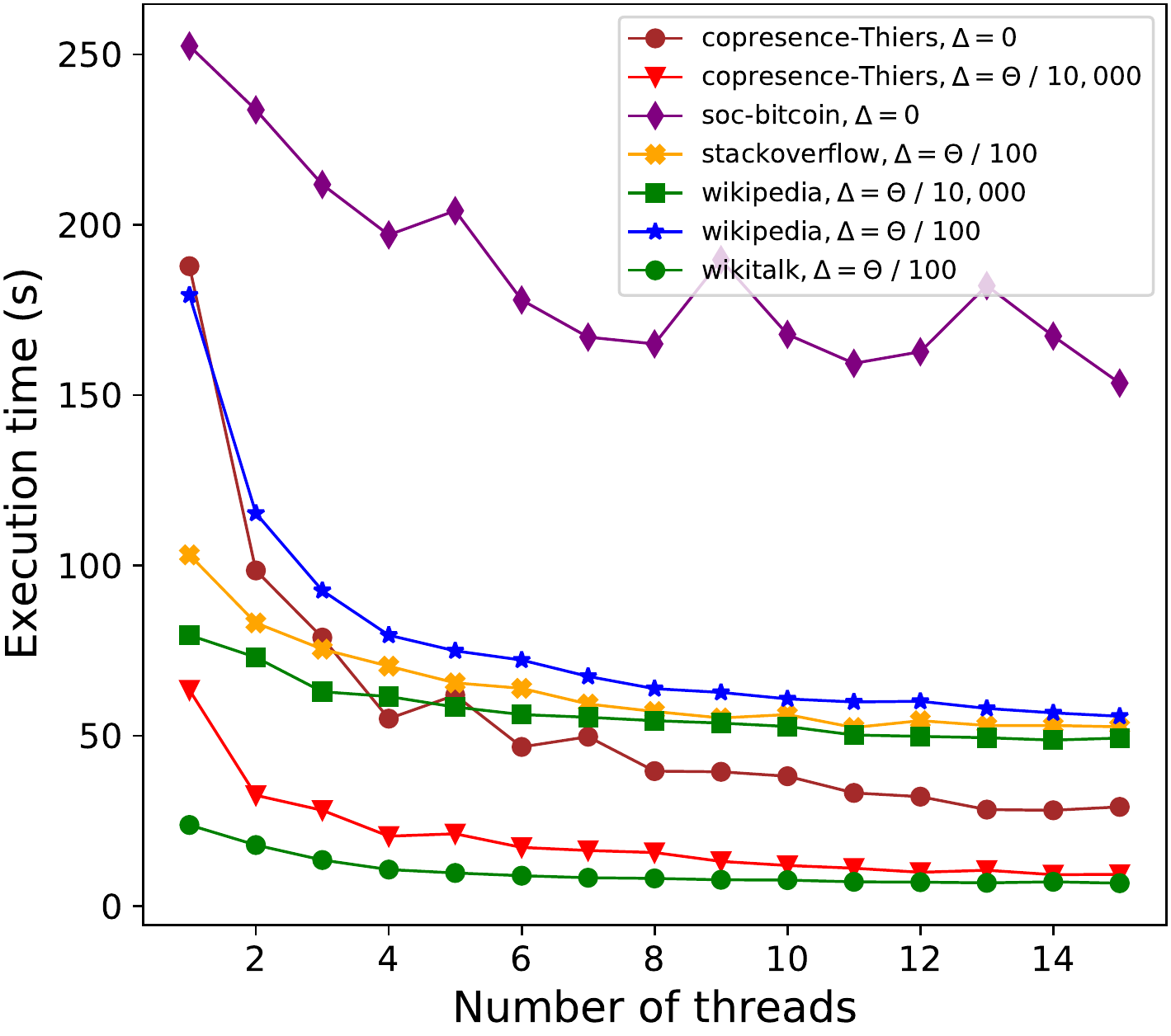}
  \hfill
  \includegraphics[width=0.48\linewidth]{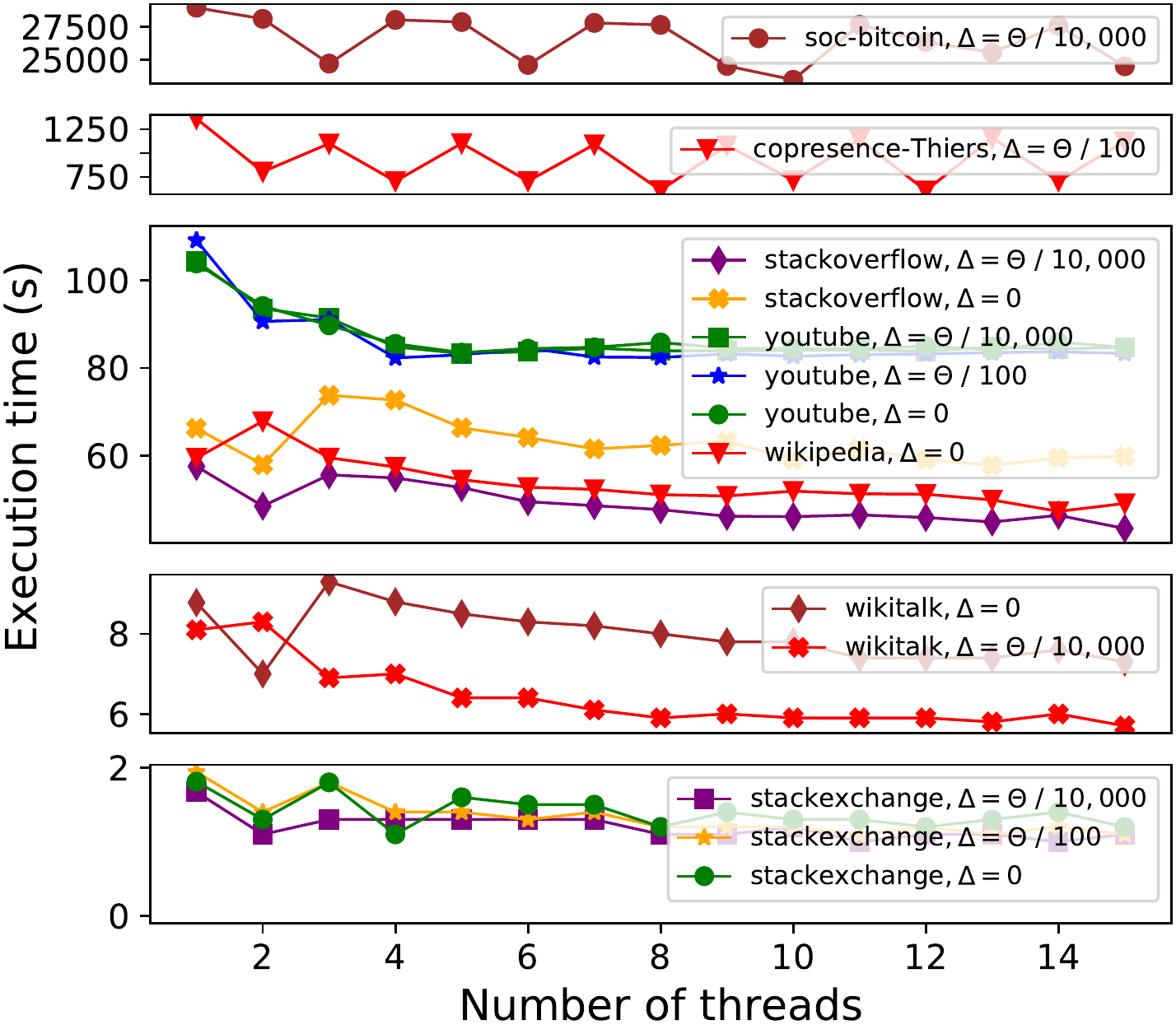} 
  
  \caption{Computation times (in seconds) as a function of the number of threads for the parallel version of the \texttt{C++} implementation on the massive link stream datasets detailed in Table~\ref{tab:data-bigLS}.
    The duration $\Delta$ of the links is expressed as a function of $\Theta$, the total duration of the link stream.
    {\bf Left:} Link streams for which parallelization offers an interesting reduction in computation time.
    {\bf Right:} Link streams for which parallelization does not work.
  }
  \label{fig:parallel}
\end{figure}

%% file: Results/data-bentert.tex
\begin{tabular}{|cc|cccc|}
  \hline % \textbf{Link duration}
  % \textbf{Link stream name} & \(\boldsymbol{\Delta}\)   & \(\boldsymbol{m}\) & \(\boldsymbol{d}\) & \(\boldsymbol{\alpha}\) & \(\boldsymbol{q}\)\\
  \textbf{Dataset} & \(\boldsymbol{\Delta}\)   & \(\boldsymbol{m}\) & \(\boldsymbol{d}\) & \(\boldsymbol{\alpha}\) & \(\boldsymbol{q}\)\\
  \hline
                              & \textbf{\emph{0}} & 8,111 & 72 & 8,111 & 2\\
  \textbf{\emph{dnc~\cite{data-dnc}}} & \textbf{\emph{125}} & 7,967 & 99 & 8,038 & 4\\
                              & \textbf{\emph{3125}} & 7,454 & 108 & 7,834 & 4\\
  \hline
                              & \textbf{\emph{0}} & 20,818 & 9 & 19,037 & 6\\
  \textbf{\emph{hypertext~\cite{data-infectious-hypertext}}} & \textbf{\emph{125}} & 6,323 & 14 & 6,859 & 7\\
                              & \textbf{\emph{3125}} & 4,082 & 48 & 6,308 & 7\\
  \hline
                              & \textbf{\emph{0}} & 28,539 & 8 & 26,384 & 5\\
  \textbf{\emph{highschool11~\cite{data-highschool2011}}} & \textbf{\emph{125}} & 6,472 & 19 & 7,732 & 7\\
                              & \textbf{\emph{3125}} & 3,636 & 34 & 7,500 & 10\\
  \hline
                              & \textbf{\emph{0}} & 32,424 & 7 & 27,835 & 5\\
  \textbf{\emph{hospital-ward~\cite{data-hospital-ward}}} & \textbf{\emph{125}} & 7,971 & 12 & 9,731 & 6\\
                              & \textbf{\emph{3125}} & 3,033 & 25 & 9,856 & 9\\
  \hline
                              & \textbf{\emph{0}} & 45,047 & 5 & 42,105 & 5\\
  \textbf{\emph{highschool12~\cite{data-highschool2011}}} & \textbf{\emph{125}} & 11,329 & 10 & 12,115 & 5\\
                              & \textbf{\emph{3125}} & 5,691 & 18 & 7,268 & 7\\
  \hline
                              & \textbf{\emph{0}} & 59,795 & 31 & 59,795 & 2\\
  \textbf{\emph{facebooklike~\cite{data-facebooklike}}} & \textbf{\emph{125}} & 50,056 & 78 & 50,080 & 3\\
                              & \textbf{\emph{3125}} & 34,116 & 92 & 34,342 & 4\\
  \hline
                              & \textbf{\emph{0}} & 110,581 & 1,458 & 97,687 & 10\\
  \textbf{\emph{as-733~\cite{data-as-733}}} & \textbf{\emph{125}} & 32,485 & 1,568 & 41,965 & 12\\
                              & \textbf{\emph{3125}} & 21,466 & 1,851 & 49,293 & 18\\
  \hline
\end{tabular}
\begin{tabular}{|cc|cccc|}
  \hline % \textbf{Link duration}
  % \textbf{Link stream name} & \(\boldsymbol{\Delta}\)   & \(\boldsymbol{m}\) & \(\boldsymbol{d}\) & \(\boldsymbol{\alpha}\) & \(\boldsymbol{q}\)\\
  \textbf{Dataset} & \(\boldsymbol{\Delta}\)   & \(\boldsymbol{m}\) & \(\boldsymbol{d}\) & \(\boldsymbol{\alpha}\) & \(\boldsymbol{q}\)\\
  \hline  
                              & \textbf{\emph{0}} & 125,773 & 4 & 106,879 & 5\\
  \textbf{\emph{primaryschool~\cite{data-primary-school}}} & \textbf{\emph{125}} & 49,530 & 16 & 67,820 & 6\\
                              & \textbf{\emph{3125}} & 19,513 & 50 & 194,231 & 14\\
  \hline
                              & \textbf{\emph{0}} & 188,508 & 4 & 172,035 & 5\\
  \textbf{\emph{highschool13~\cite{data-highschool2013}}} & \textbf{\emph{125}} & 36,277 & 14 & 41,534 & 6\\
                              & \textbf{\emph{3125}} & 15,764 & 30 & 28,357 & 8\\
  \hline
                              & \textbf{\emph{0}} & 312,164 & 7 & 294,269 & 3\\
  \textbf{\emph{london~\cite{data-transport}}} & \textbf{\emph{125}} & 27,595 & 7 & 28,683 & 3\\
                              & \textbf{\emph{3125}} & 768 & 7 & 778 & 3\\
  \hline
                              & \textbf{\emph{0}} & 353,226 & 8 & 342,540 & 3\\
  \textbf{\emph{paris~\cite{data-transport}}} & \textbf{\emph{125}} & 50,248 & 8 & 51,084 & 3\\
                              & \textbf{\emph{3125}} & 1,080 & 8 & 1,093 & 3\\
  \hline
                              & \textbf{\emph{0}} & 415,000 & 134 & 229,144 & 19\\
  \textbf{\emph{flights~\cite{data-transport}}} & \textbf{\emph{125}} & 415,000 & 134 & 229,144 & 19\\
                              & \textbf{\emph{3125}} & 37,862 & 139 & 368,756 & 19\\
  \hline
                              & \textbf{\emph{0}} & 415,912 & 4 & 338,815 & 5\\
  \textbf{\emph{infectious~\cite{data-infectious-hypertext}}} & \textbf{\emph{125}} & 100,329 & 15 & 138,670 & 7\\
                              & \textbf{\emph{3125}} & 44,767 & 43 & 150,883 & 16\\
  \hline
                              & \textbf{\emph{0}} & 468,897 & 12 & 442,341 & 3\\
  \textbf{\emph{ny~\cite{data-transport}}} & \textbf{\emph{125}} & 113,651 & 12 & 117,481 & 3\\
                              & \textbf{\emph{3125}} & 661 & 12 & 674 & 3\\
  \hline
\end{tabular}

%% file: Results/data-bigLS.tex
\begin{tabular}{|cc|cccc|}
  \hline
  \textbf{Dataset} & \(\boldsymbol{\Delta}\) & \(\boldsymbol{m}\) & \(\boldsymbol{d}\) & \(\boldsymbol{\alpha}\) & \(\boldsymbol{q}\)\\
  \hline
                   & \textbf{\emph{0}} & 1,108,715 & 3 & 1,108,715 & 2\\
  \textbf{\emph{stackexchange~\cite{data-stackoverflow}}} & \textbf{\emph{23,961}} & 870,128 & 47 & 894,317 & 5\\
                   & \textbf{\emph{2,396,149}} & 751,974 & 671 & 1,030,023 & 10\\
  \hline
                   & \textbf{\emph{0}} & 6,092,321 & 22 & 6,092,321 & 2\\
  \textbf{\emph{wikitalk~\cite{data-stackoverflow}}} & \textbf{\emph{20,048}} & 4,123,960 & 12,205 & 4,207,362 & 7\\
                   & \textbf{\emph{2,004,838}} & 3,078,861 & 29,543 & 4,500,754 & 10\\
  \hline
                   & \textbf{\emph{0}} & 12,223,774 & 28,714 & 12,253,571 & 17\\
  \textbf{\emph{youtube~\cite{data-youtube}}} & \textbf{\emph{1,944}} & 12,223,774 & 28,714 & 12,253,571 & 17\\
                   & \textbf{\emph{194,400}} & 10,310,419 & 28,714 & 10,656,065 & 17\\
  \hline
                   & \textbf{\emph{0}} & 18,613,039 & 79 & 131,251 & 80\\
  \textbf{\emph{copresence-Thiers~\cite{data-soc-bitcoin}}} & \textbf{\emph{38}} & 3,857,645 & 100 & 412,553 & 80\\
                   & \textbf{\emph{3,804}} & 147,125 & 217 & 5,572,145 & 102\\
  \hline
                   & \textbf{\emph{0}} & 39,949,279 & 3,463 & 39,935,611 & 13\\
  \textbf{\emph{wikipedia~\cite{data-youtube}}} & \textbf{\emph{19,318}} & 39,246,821 & 7,216 & 40,898,684 & 28\\
                   & \textbf{\emph{1,931,869}} & 38,737,308 & 36,121 & 45,440,988 & 28\\
  \hline
                   & \textbf{\emph{0}} & 47,902,566 & 4 & 47,902,566 & 2\\
  \textbf{\emph{stackoverflow~\cite{data-stackoverflow}}} & \textbf{\emph{23,970}} & 33,948,538 & 90 & 35,298,283 & 6\\
                   & \textbf{\emph{2,397,055}} & 29,583,489 & 1,443 & 43,989,692 & 13\\
  \hline
                   & \textbf{\emph{0}} & 122,378,012 & 3,769 & 119,172,078 & 219\\
  \textbf{\emph{soc-bitcoin~\cite{data-soc-bitcoin}}} & \textbf{\emph{15,653}} & 93,897,987 & 28,144 & 787,519,128 & 237\\
                   & \textbf{\emph{1,565,366}} & 86,668,193 & 170,760 & - & -\\
  \hline
\end{tabular}

%% file: Results/time-all-bentert.tex
  
\begin{tabular}{|cc|ccccc|}
  \hline
  \textbf{Dataset} & \textbf{$ \Delta $} & \textbf{C++} & \textbf{Python} & \textbf{BHM+} & \textbf{VML} & \textbf{HMNS}\\
  \hline
                   & \textbf{\emph{0}} & 0.05 & 0.08 & 24 & 1.5 & 178\\
  \textbf{\emph{dnc}} & \textbf{\emph{125}} & 0.05 & 0.08 & 24 & 2.6 & 155\\
                   & \textbf{\emph{3,125}} & 0.04 & 0.08 & 23 & 2.7 & 80\\
  \hline
                   & \textbf{\emph{0}} & 0.06 & 0.16 & 2.6 & 2.8 & 65\\
  \textbf{\emph{hypertext}} & \textbf{\emph{125}} & 0.05 & 0.08 & 1.4 & 1.8 & 6.9\\
                   & \textbf{\emph{3,125}} & 0.03 & 0.09 & 1.4 & 2.0 & 3.6\\
  \hline
                   & \textbf{\emph{0}} & 0.08 & 0.20 & 2.4 & 2.6 & 95\\
  \textbf{\emph{highschool11}} & \textbf{\emph{125}} & 0.05 & 0.09 & 1.2 & 1.6 & 5.4\\
                   & \textbf{\emph{3,125}} & 0.03 & 0.16 & 2.0 & 4.1 & 3.1\\
  \hline
                   & \textbf{\emph{0}} & 0.08 & 0.25 & 3.7 & 4.3 & 271\\
  \textbf{\emph{hospital-ward}} & \textbf{\emph{125}} & 0.02 & 0.12 & 1.6 & 2.6 & 17\\
                   & \textbf{\emph{3,125}} & 0.04 & 0.25 & 3.5 & 7.8 & 5.0\\
  \hline
                   & \textbf{\emph{0}} & 0.12 & 0.32 & 3.6 & 4.2 & 183\\
  \textbf{\emph{highschool12}} & \textbf{\emph{125}} & 0.06 & 0.12 & 1.6 & 2.1 & 12\\
                   & \textbf{\emph{3,125}} & 0.03 & 0.09 & 1.3 & 1.5 & 3.8\\
  \hline
                   & \textbf{\emph{0}} & 0.11 & 0.39 & 27 & 10 & 129\\
  \textbf{\emph{facebooklike}} & \textbf{\emph{125}} & 0.12 & 0.34 & 26 & 15 & 96\\
                   & \textbf{\emph{3,125}} & 0.35 & 0.25 & 25 & 10 & 59\\
  \hline
                   & \textbf{\emph{0}} & 0.31 & 1.5 & 569 & 392 & -\\
  \textbf{\emph{as-733}} & \textbf{\emph{125}} & 0.21 & 1.5 & 368 & 1,581 & 528\\
                   & \textbf{\emph{3,125}} & 0.11 & 5.1 & 562 & \(\times\) & 572\\
  \hline
\end{tabular}  
\begin{tabular}{|cc|ccccc|}
  \hline
  \textbf{Dataset} & \textbf{$ \Delta $} & \textbf{C++} & \textbf{Python} & \textbf{BHM+} & \textbf{VML} & \textbf{HMNS}\\
  \hline
                   & \textbf{\emph{0}} & 0.53 & 0.87 & 22 & 30 & 716\\
  \textbf{\emph{primaryschool}} & \textbf{\emph{125}} & 0.14 & 0.65 & 15 & 37 & 125\\
                   & \textbf{\emph{3,125}} & 0.34 & 10 & 204 & 854 & 117\\
  \hline
                   & \textbf{\emph{0}} & 0.30 & 1.2 & 19 & 24 & 1,420\\
  \textbf{\emph{highschool13}} & \textbf{\emph{125}} & 0.11 & 0.37 & 6.3 & 11 & 54\\
                   & \textbf{\emph{3,125}} & 0.08 & 0.38 & 6.3 & 11 & 14\\
  \hline
                   & \textbf{\emph{0}} & 0.40 & 1.8 & 8.6 & 11 & 3,712\\
  \textbf{\emph{london}} & \textbf{\emph{125}} & 0.10 & 0.20 & 2.1 & 1.3 & 45\\
                   & \textbf{\emph{3,125}} & 0.01 & 0.04 & 1.5 & 0.06 & 1.6\\
  \hline
                   & \textbf{\emph{0}} & 0.44 & 2.0 & 10 & 13 & 4,260\\
  \textbf{\emph{paris}} & \textbf{\emph{125}} & 0.12 & 0.33 & 2.8 & 2.5 & 101\\
                   & \textbf{\emph{3,125}} & 0.02 & 0.04 & 1.7 & 0.07 & 1.8\\
  \hline
                   & \textbf{\emph{0}} & 1.1 & 3.1 & 665 & 22 & 1,206\\
  \textbf{\emph{infectious}} & \textbf{\emph{125}} & 0.78 & 1.6 & 634 & 20 & 945\\
                   & \textbf{\emph{3,125}} & 1.1 & 8.1 & 818 & 534 & 1,004\\
  \hline
                   & \textbf{\emph{0}} & 2.7 & 93 & 36,859 & \(\times\) & 26,109\\
  \textbf{\emph{flights}} & \textbf{\emph{125}} & 2.6 & 93 & 37,076 & \(\times\) & 26,411\\
                   & \textbf{\emph{3,125}} & 1.4 & 89 & 13,420 & \(\times\) & 2,555\\
  \hline
                   & \textbf{\emph{0}} & 0.56 & 2.8 & 13 & 18 & 6,084\\
  \textbf{\emph{ny}} & \textbf{\emph{125}} & 0.18 & 0.75 & 5.2 & 6.9 & 304\\
                   & \textbf{\emph{3,125}} & 0.02 & 0.03 & 2.3 & 0.06 & 2.5\\
  \hline

\end{tabular}

%% file: Results/time-all-bigLS.tex
\begin{tabular}{|cc|ccccc|}
  \hline
  \textbf{Dataset} & $ \Delta $ & \textbf{C++} & \textbf{Python} & \textbf{BHM+} & \textbf{VML} & \textbf{HMNS}\\
  \hline
                   & \textbf{\emph{0}} & 1.8 & 7.5 & - & 2,576 & -\\
  \textbf{\emph{stackexchange}} & \textbf{\emph{23,961}} & 1.7 & 6.4 & - & 3,747 & -\\
                   & \textbf{\emph{2,396,149}} & 1.9 & 10 & - & \(\times\) & -\\
  \hline
                   & \textbf{\emph{0}} & 8.8 & 42 & - & \(\times\) & -\\
  \textbf{\emph{wikitalk}} & \textbf{\emph{20,048}} & 8.1 & 33 & - & \(\times\) & -\\
                   & \textbf{\emph{2,004,838}} & 23 & 53 & - & \(\times\) & -\\
  \hline
                   & \textbf{\emph{0}} & 103 & 459 & - & \(\times\) & -\\
  \textbf{\emph{youtube}} & \textbf{\emph{1,944}} & 104 & 461 & - & \(\times\) & -\\
                   & \textbf{\emph{194,400}} & 109 & 408 & - & \(\times\) & -\\
  \hline
                   & \textbf{\emph{0}} & 185 & 19,525 & - & \(\times\) & -\\
  \textbf{\emph{copresence-Thiers}} & \textbf{\emph{38}} & 64 & 16,305 & - & \(\times\) & -\\
                   & \textbf{\emph{3,804}} & 1,347 & - & - & \(\times\) & -\\
  \hline
                   & \textbf{\emph{0}} & 59 & 322 & - & \(\times\) & -\\
  \textbf{\emph{wikipedia}} & \textbf{\emph{19,318}} & 79 & 403 & - & \(\times\) & -\\
                   & \textbf{\emph{1,931,869}} & 179 & 547 & - & \(\times\) & -\\
  \hline
                   & \textbf{\emph{0}} & 66 & 343 & - & \(\times\) & -\\
  \textbf{\emph{stackoverflow}} & \textbf{\emph{23,970}} & 57 & 302 & - & \(\times\) & -\\
                   & \textbf{\emph{2,397,055}} & 103 & 568 & - & \(\times\) & -\\
  \hline
                   & \textbf{\emph{0}} & 249 & 6,505 & - & \(\times\) & -\\
  \textbf{\emph{soc-bitcoin}} & \textbf{\emph{15,653}} & 27,660 & - & - & \(\times\) & -\\
                   & \textbf{\emph{1,565,366}} & - & - & - & \(\times\) & -\\
  \hline
\end{tabular}

%% file: Results/pivoting.tex
\begin{tabular}{|cc|ccc|ccc|}
  \hline
  \multicolumn{2}{|c|}{\textbf{Link stream}} & \multicolumn{3}{c|}{\textbf{With pivot}} & \multicolumn{3}{c|}{\textbf{Without pivot}} \\ \hline
  \textbf{Dataset} & \(\boldsymbol{\Delta}\) & \textbf{\(\boldsymbol{t}\)} & \(\boldsymbol{r}\) & \(\boldsymbol{1/r}\) & \textbf{\(\boldsymbol{t}\)} & \(\boldsymbol{r}\) & \(\boldsymbol{1/r}\)\\
  \hline
                                             & \textbf{\emph{0}} & 1.8 & 1.000 & 1 & 1.8 & 1.000 & 1\\
  \textbf{\emph{stackexchange}} & \textbf{\emph{23,961}} & 1.7 & 1.000 & 1 & 1.7 & 1.000 & 1\\
                                             & \textbf{\emph{2,396,149}} & 1.9 & 0.997 & 1.003 & 1.7 & 0.928 & 1.078\\
  \hline
                                             & \textbf{\emph{0}} & 8.8 & 1.000 & 1 & 7.7 & 1.000 & 1\\
  \textbf{\emph{wikitalk}} & \textbf{\emph{20,048}} & 8.1 & 1.000 & 1 & 7.7 & 1.000 & 1\\
                                             & \textbf{\emph{2,004,838}} & 23 & 0.995 & 1.005 & 23 & 0.891 & 1.122\\
  \hline
                                             & \textbf{\emph{0}} & 103 & 0.907 & 1.103 & 135 & 0.332 & 3.012\\
  \textbf{\emph{youtube}} & \textbf{\emph{1,944}} & 104 & 0.907 & 1.103 & 135 & 0.332 & 3.012\\
                                             & \textbf{\emph{194,400}} & 109 & 0.900 & 1.111 & 137 & 0.298 & 3.356\\
  \hline
                                             & \textbf{\emph{0}} & 185 & 0.087 & 11.49 & - & - & -\\
  \textbf{\emph{copresence-Thiers}} & \textbf{\emph{38}} & 64 & 0.757 & 1.321 & - & - & -\\
                                             & \textbf{\emph{3,804}} & 1,347 & 0.854 & 1.171 & - & - & -\\
  \hline
                                             & \textbf{\emph{0}} & 59 & 1.000 & 1 & 78 & 0.987 & 1.013\\
  \textbf{\emph{wikipedia}} & \textbf{\emph{19,318}} & 79 & 1.000 & 1 & 1,143 & 0.043 & 23.26\\
                                             & \textbf{\emph{1,931,869}} & 179 & 0.999 & 1.001 & 1,618 & 0.034 & 29.41\\
  \hline
                                             & \textbf{\emph{0}} & 66 & 1.000 & 1 & 90 & 1.000 & 1\\
  \textbf{\emph{stackoverflow}} & \textbf{\emph{23,970}} & 57 & 1.000 & 1 & 73 & 1.000 & 1\\
                                             & \textbf{\emph{2,397,055}} & 103 & 0.996 & 1.004 & 122 & 0.881 & 1.135\\
  \hline
                                             & \textbf{\emph{0}} & 249 & 0.972 & 1.029 & - & - & -\\
  \textbf{\emph{soc-bitcoin}} & \textbf{\emph{15,653}} & 27,660 & 0.911 & 1.098 & - & - & -\\
                                             & \textbf{\emph{15,653,366}} & - & - & - & - & - & -\\
  \hline
\end{tabular}

%% file: conclusion.tex
\section{Conclusion}
\label{conclusion}

% Summary
% 
In this paper, we have addressed the problem of maximal clique enumeration in link streams.
We propose a new algorithm to solve this problem that scales to massive real-world link streams.
We analyze its complexity  as a function of the characteristics of the input and as a function of the characteristics of the output of the algorithm
We provide two implementations (in {\tt Python} and in {\tt C++}) and perform an experimental protocol on various datasets from real interactions over time. It shows that our algorithm allows a performance gain of several orders of magnitude with respect to the state of the art.

% Direct perspectives
% 
The work done in this paper can be pursued along several directions.
One is to investigate the reasons why the enumeration takes so long on some link streams, especially in the cases where the time limit of our protocol is reached.
For instance, as the cliques are output on the fly, we could compute the $\frac{1}{r}$ factor appearing in the output-sensitive complexity 
on the fly and find out whether the computation stalls because it explores many unnecessary branches, or if it comes from the sheer number of maximal cliques.
Another direction would be to improve the parallelization process by refining the partitioning of the link stream based on its structural properties and balance the computation more fairly on the different threads.
For instance, we can study the structure of the link stream to anticipate the instants $t \in T$ where there are more or less maximal cliques, which would allow the construction of more suitable sub-intervals and thus a better speedup.
Also, it would be interesting to study the impact of the link duration on the enumeration of maximal cliques, both in terms of computation times and number and size of cliques output.

% Broad perspectives
% 
More broadly, listing other kinds of dense sub-streams in link streams is a relevant problem.
In this direction, we proposed in another work~\cite{baudin2023lscpm} an algorithm to extend to the context of links streams the well-known clique percolation method which defines communities in graphs~\cite{palla2005uncovering}.  
It demands to list $k$-cliques in link streams, which is a different problem from the one of listing maximal cliques (and computationally more tractable for small values of $k$).
We showed that using fast $k$-clique enumeration processes allows obtaining communities more efficiently than the other extension to the dynamical context of the clique percolation method~\cite{boudebza2018olcpm}.
As other types of dense sub-streams have been recently proposed~\cite{bentert2019,molter2021isolation,banerjee2022efficient}, it would be interesting to generalize our enumeration method to these cases.